\documentclass{amsart}
\usepackage{amsmath,amssymb,amsthm,latexsym,amsfonts}
\usepackage{palatino,fancyhdr,graphicx,hyperref}
\usepackage{color}
\usepackage{verbatim}
\usepackage{tikz-cd}
\usepackage{setspace}
\usepackage{adjustbox}

\def \CC {{\mathbb C}}
\def \RR {{\mathbb R}}

\def \p {{\mathfrak {p}}}
\def \O {{\mathcal O}}
\def \D {{\mathcal{D}}}
\def \P {{\mathfrak{P}}}

\def \uv {{\boldsymbol u}}
\def \vv {{\boldsymbol v}}
\def \xv {{\boldsymbol x}}
\def \yv {{\boldsymbol y}}
\def \cv {{\boldsymbol c}}
\def \sv {{\boldsymbol s}}
\def \tv {{\boldsymbol t}}
\def \wv {{\boldsymbol \omega}}

\def \re {{{\rm Re}}}
\def \im {{{\rm Im}}}

\newcommand\Tr[1] {{\rm{Tr}}\left(#1\right)}

\newcommand\de[1] {{\rm{det}}\left(#1\right)}

\newcommand{\FF}{\mathbb{F}}
\newcommand{\QQ}{\mathbb{Q}}
\newcommand{\ZZ}{\mathbb{Z}}
\newcommand{\Oc}{\mathcal{O}}
\newcommand{\pf}{\mathfrak{p}}

\theoremstyle{definition}
\newtheorem{thm}{Theorem}[section]
\newtheorem{lem}[thm]{Lemma}

\newtheorem{definition}[thm]{Definition}

\newtheorem{lemma}[thm]{Lemma}
\newtheorem{proposition}[thm]{Proposition}

\newtheorem{corollary}[thm]{Corollary}
\newtheorem{remark}[thm]{Remark}

\newtheorem{example}[thm]{Example}

\title{Modular Lattices from a Variation of Construction A over Number Fields}
\author{Xiaolu Hou and Fr\'ed\'erique Oggier}
\address{Division of Mathematical Sciences, School of Physical and Mathematical Sciences, \newline\indent Nanyang Technological University, Singapore.} 
\email {HO0001LU@e.ntu.edu.sg,frederique@ntu.edu.sg}
\thanks{X. Hou is supported by a Nanyang President Graduate Scholarship. The research of F. Oggier for this work is supported by Nanyang Technological University under Research Grant M58110049.}

\begin{document}
\begin{abstract}
We consider a variation of Construction A of lattices from linear codes based on two classes of number fields, totally real and CM Galois number fields. We propose a generic construction with explicit generator and Gram matrices, then focus on modular and unimodular lattices, obtained in the particular cases of totally real, respectively, imaginary, quadratic fields.
Our motivation comes from coding theory, thus some relevant properties of modular lattices, such as minimal norm, theta series, kissing number and secrecy gain are analyzed. Interesting lattices are exhibited.
\end{abstract}
\maketitle
%
%
\section{Introduction}

Let $K$ be a Galois number field of degree $n$ which is either totally real or a CM field. 
Let $\mathcal{O}_K$ be the ring of integers of $K$ and $\mathfrak{p}$ be a prime ideal of $\O_K$ above the prime $p$.
We have $\O_K/\mathfrak{p}\cong\FF_{p^f}$, where $f$ is the inertia degree of $p$.
Define $\rho$ to be the map of reduction modulo $\mathfrak{p}$ componentwise as follows:
\begin{equation}\label{eqn:maprho}
\begin{array}{cccc}
\rho :& \O_K^N &\rightarrow& \FF_{p^f}^N\\
    ~ &(x_1,\ldots,x_N) &\mapsto& (x_1~{\rm mod}~\mathfrak{p},\ldots,x_N~{\rm mod}~\mathfrak{p})
\end{array}
\end{equation}
for some positive integer $N$.
Let $C\subseteq\FF_{p^f}^N$ be a linear code over $\FF_{p^f}^N$, that is a $k$-dimensional subspace of $\FF_{p^f}^N$.
As $\rho$ is a $\ZZ-$module homomorphism, $\rho^{-1}(C)$ is a submodule of $\O_K^N$.
Since $\O_K$ is a free $\ZZ-$module of rank $n$, $\rho^{-1}(C)$ is a free $\ZZ-$module of rank $nN$. 
Let $b_\alpha:\O_K^N\times\O_K^N\to\RR$ be the symmetric bilinear form defined by
\begin{equation}\label{eqn:bilinearform}
b_\alpha(\xv,\yv)=\sum_{i=1}^N{\rm Tr}_{K/\QQ}(\alpha x_i \bar{y}_i)
\end{equation}
where $\alpha \in K\cap \RR$ and $\bar{y}_i$ denotes the complex conjugate of $y_i$ if $K$ is CM (and $\bar{y}_i$ is understood to be $y_i$ if $K$ is totally real).
If $\alpha$ is furthermore totally positive, i.e., $\sigma_i(\alpha)>0$, for $\sigma_1 (\text{the identity}),\sigma_2,\ldots,\sigma_n$ all elements of the Galois group of $K$ over $\QQ$, then 
$b$ is positive definite: 
\[
b_\alpha(\xv,\xv)=\sum_{i=1}^N\Tr{\alpha x_i\bar{x}_i}=\sum_{i=1}^N\sum_{j=1}^n\sigma_j(\alpha)|\sigma_j(x_i)|^2>0,
\] 
$\forall \xv\in\O_K^N$, $\xv$ nonzero.
If we take $\alpha$ in the codifferent $\D_K^{-1}=\{x\in K:\Tr{xy}\in\ZZ~\forall y\in\O_K \}$ of $K$, then $\Tr{\alpha x_i\bar{y}_i}\in\ZZ$. 

The pair $(\rho^{-1}(C),b_\alpha)$ thus forms a lattice of rank (or dimension) $nN$, which is integral when $\alpha\in\D_K^{-1}$ 
but also in other cases, depending on the choice of $C$, as we will see several times next.

This method of constructing lattices from linear codes is usually referred to as Construction $A$~\cite{Conway}. The principle is well known, albeit not using the exact above formulation. The original binary Construction A, due to Forney~\cite{Forney}, uses $K=\QQ$, $\Oc_K=\ZZ$, $\p=2$ and typically $\alpha=1/2$ (sometimes $\alpha$ is chosen to be 1).
The binary Construction A can also be seen as a particular case of the cyclotomic field approach proposed by Ebeling~\cite{Ebeling}, which in turn is a particular case of the above construction.
For $p$ a prime, take for $K$ the cyclotomic field $\QQ(\zeta_p)$, where $\zeta_p$ is a primitive $p$th root of unity, and note that $\O_K=\ZZ[\zeta_p]$.  
Take $\p=(1-\zeta_p)$ the prime ideal above $p$, and $\alpha=1/p$. Since $\O_K/\p\cong\FF_p$, this construction involves linear codes over $\FF_p$. The case $p=2$ is the binary Construction A.
The generalization from cyclotomic fields to either CM fields or totally real number fields was suggested in~\cite{Ong} for the case where $\pf$ is totally ramified. The motivation to revisit Construction A using number fields came from coding theory for wireless communication, for which lattices built over totally real numbers fields and CM fields play an important role~\cite{FnT}. In particular, Construction A over number fields enables lattice coset encoding for transmission over wireless channels, and wireless wiretap channels \cite{Ong}. It is also useful in the context of physical network coding~\cite{BO13}.

The main interest in constructing lattices from linear codes is to take advantage of the code properties to obtain lattices with nice properties, modularity and large shortest vector (or minimal norm) being two of them, both mathematically but also for coding applications.

Given an arbitrary lattice $(L,b)$ where $L$ is a $\ZZ-$module and $b$ is a symmetric bilinear form which is positive definite, then the dual lattice of 
$(L,b)$ is the pair $(L^*,b)$, where
\[
	L^* = \{\xv\in L\otimes_\ZZ\RR:b(\xv,\yv)\in\ZZ \ \forall \yv\in L\},
\]
and $(L,b)$ is 

$\bullet$ integral if $L\subseteq L^*$,

$\bullet$ unimodular if $(L,b)\cong (L^*,b)$, i.e., there exists a $\ZZ-$module homomorphism $\tau:L\to L^*$ such that $b(\tau(x),\tau(y))=b(x,y)$ for all $x,y\in L$, and

$\bullet$ {\em $d-$modular} (or {\em modular of level} $d$) if it is integral and $(L,b)\cong (L^*,d b)$ for some positive integer $d$.

Given a linear code $C \subset \FF_q^N$ of dimension $k$, $q$ a prime power, its dual code $C^\perp$ is defined by
\begin{equation}\label{eqn:dualcode}
C^\perp = \{\xv\in\FF_q^N:\xv\cdot\yv=\sum_{i=1}^N x_iy_i=0\ \forall \yv\in C\}
\end{equation}
and $C$ is called 

$\bullet$ {\em self-orthogonal} if $C\subseteq C^\perp$, and 

$\bullet$ {\em self-dual} if $C=C^\perp$.

It is well known for the binary Construction A that $C\subseteq \FF_2^N$ is self-dual if and only if $(\rho^{-1}(C),b_{\frac{1}{2}})$ is unimodular~\cite{Ebeling,Conway}. 
More generally, for $K=\QQ(\zeta_p)$, if $C\subseteq\FF_{p}^N$ is self-dual, then $(\rho^{-1}(C),b_{\frac{1}{p}})$ is unimodular~\cite{Ebeling}. We will prove a converse of this statement for totally real number fields and CM fields with a totally ramified prime in Section \ref{sec:matrix}. 

Self-dual codes thus provide a systematic way to obtain modular lattices. This was used for example in \cite{Chapman}, where $K=\QQ(\sqrt{-2})$, $\p=(3)$ and self-dual codes over the ring $\Oc_K/\p$ were used to construct $2-$modular lattices. 
Similarly, in \cite{Chua}, it was shown that by taking $K=\QQ(\zeta_3)$, where $\zeta_3$ is the $3$rd primitive root of unity, $\p=(4)$, and self-dual codes over the ring $\O_K/\p$, $3-$modular lattices can be constructed.
In \cite{Bachoc}, the quadratic fields $\QQ(\sqrt{-7})$ with $\p=(2)$, $\QQ(i)$ with $\p=(2)$ and $\QQ(\zeta_3)$ with $\p=(2)$ or $\p=(3)$, as well as totally definite quaternion algebras ramified at either 2  or 3 with $\p=(2)$, were used to construct modular lattices from self-dual codes.
An even more generalized version of Construction A is introduced in \cite{Sloane}, where $\Oc_K$ is replaced by any lattice $L\subset\RR^n$ and $\p$ by $pL$ for a prime $p$.
It is then applied to construct unimodular lattices from self-dual linear codes.

Apart from modularity, large minimal norm is another classical property which has been well studied.
This is normally achieved via Construcion A by exploiting the dualities between the linear codes and the resulting lattices.
For example, in \cite{Bachoc}, the association between MacWilliams identities for linear codes and theta series for lattices are established for the cases listed above to construct extremal lattices, lattices with the largest possible minimal norm.
Other duality relations also include the relation between the minimum weight of linear codes and the minimal norm of the corresponding lattices~\cite{Chapman}, or the connection between the weight enumerator of linear codes and the theta series of lattices~\cite{Chua}, shown in both cases for the cases listed above. 
One classical motivation for finding lattices with the biggest minimum is to find the densest sphere packings, which can be applied to coding over Gaussian channels~\cite{Sloane}.

In Section \ref{sec:matrix}, generator and Gram matrices are computed for the generic case of Construction A over Galois number fields, either totally real or CM. Knowledge of these matrices is important for applications, such as lattice encoding, or if one needs to compute the theta series of the lattice, as we will do in Section \ref{sec:realproperties}. It also gives one way to verify modularity, as will be shown both in Section \ref{sec:real} and \ref{sec:imaginary}. 
From the generic construction, examples of lattices are obtained by considering specific number fields. We investigate the two most natural ones, namely totally real quadratic fields in Section \ref{sec:real}, and totally imaginary quadratic fields in Section \ref{sec:imaginary}. Our techniques could be applied to other number fields, such as cyclotomic fields, or cyclic fields, but these directions are left open.
Section \ref{sec:realproperties} provides examples of lattices and of their applications: we construct modular lattices and compute their theta series (and their kissing number in particular), but also their minimal norm. The theta series allows to compute the secrecy gain of the lattice~\cite{Oggier}, a lattice invariant studied in the context of wiretap coding. Interesting examples are found -- a new extremal lattice or new constructions of known extremal lattices, modular lattices with large minimal norm -- and numerical evidence gives new insight on the behaviour of the secrecy gain.

%
%
%

\section{Generator and Gram Matrices for Construction A}\label{sec:matrix}

As above, we consider the $nN$-dimensional lattice $(\rho^{-1}(C),b_{\alpha})$. 
Let $\Delta$ be the absolute value of the discriminant of $K$. We will adopt 
the row convention, meaning that a lattice generator matrix contains a basis as row vectors. The Gram matrix contains as usual the inner product between the basis vectors. The volume of a lattice is the absolute value of the determinant of a generator matrix, while the discriminant of a lattice is the determinant of its Gram matrix.

\begin{lemma}\label{lem:latdisc}
	The lattice $(\rho^{-1}(C),b_\alpha)$ has discriminant $\Delta^Np^{2f(N-k)}N(\alpha)^N$ and volume $\Delta^{N/2}p^{f(N-k)}N(\alpha)^\frac{N}{2}$.
\end{lemma}
\begin{proof}
	For $N=1$, $(\O_K,b_\alpha)$ is a lattice with discriminant $N(\alpha)\Delta$~\cite{BayerIdeallattice}.
Hence $(\O_K^N,b_\alpha)$ has discriminant $(N(\alpha)\Delta)^N$ and volume $(N(\alpha)\Delta)^{\frac{N}{2}}$.
As $\rho$ is a surjective $\ZZ-$module homomorphism and $C$ has index $p^{f(N-k)}$ as a subgroup of $\FF_{p^f}^N$, $\rho^{-1}(C)$ also has index $p^{f(N-k)}$ as a subgroup of $\O_K^N$ and 
we have \cite{Ebeling}
\[
vol((\rho^{-1}(C),b_\alpha))=vol((\O_K^N,b_\alpha))|\O_K^N/\rho^{-1}(C)|=N(\alpha)^{N/2}\Delta^{\frac{N}{2}}p^{f(N-k)}.
\]
\end{proof}

\begin{corollary}
The dual lattice $(\rho^{-1}(C)^*,b_\alpha)$ has discriminant $\Delta^{-N}p^{-2f(N-k)}N(\alpha)^{-N}$ and volume $\Delta^{-N/2}p^{-f(N-k)}N(\alpha)^\frac{-N}{2}$.
Also the lattice $(\rho^{-1}(C^\perp),b_\alpha)$ has discriminant $\Delta^Np^{2fk}N(\alpha)^N$ and volume $\Delta^{\frac{N}{2}}p^{fk}N(\alpha)^\frac{N}{2}$.
\end{corollary}

Let $\{v_1,\dots,v_n\}$ be a $\ZZ-$basis for $\O_K$ and let $\{\omega_1,\dots,\omega_n\}$ be a $\ZZ-$basis for $\p$. 
Suppose $C$ admits a generator matrix in the standard (systematic) form and let $A$ be a matrix such that $(I_k~(A~{\rm mod}~\p))$ is a generator matrix of $C$.

\begin{proposition}\label{prop:MCReal}
For $K$ a totally real number field of degree $n$ with Galois group $\{\sigma_1,\ldots,\sigma_n\}$, 
a generator matrix for $(\rho^{-1}(C),b_\alpha)$ is given by
\begin{equation}\label{eq:MCReal}
M_C=
\begin{bmatrix}
	I_k\otimes M& A\widetilde{\otimes}M\\
	\boldsymbol{0}_{nN-nk,nk}&I_{N-k}\otimes M_p
\end{bmatrix}(I_N \otimes D_\alpha),
\end{equation}
where $M=(\sigma_i(v_j))_{i,j=1,\ldots,n}$, $M_p=(\sigma_i(w_j))_{i,j=1,\ldots,n}$ are respectively
generator matrices for $(\Oc_K^N,b_1)$ and $(\pf^N,b_1)$, $D_\alpha$ is a diagonal matrix whose 
diagonal entries are $\sqrt{\sigma_i(\alpha)}$, $i=1,\ldots,n$, and
	\[
		A\widetilde{\otimes}M:=[\sigma_1(A_1)\otimes M_1,\dots,\sigma_n(A_1)\otimes M_n,\dots,\sigma_n(A_{N-k})\otimes M_1,\dots,\sigma_n(A_{N-k})\otimes M_n],
	\]
	where we denote the columns of the matrices $M,A$ by $M_i$, $i=1,\ldots,n$, $A_j$, $j=1,2,\dots,N-k$ and $\sigma_i$ is understood componentwise, $i=1,\ldots,n$.
\end{proposition}
\begin{proof}
Note that $\de{M}=\Delta^{\frac{1}{2}}$ is the volume of $(\O_K,b_1)$~\cite{Bayer}
and similarly, $\de{M_p}=\Delta^{\frac{1}{2}}p^f$ is the volume of $(\p,b_1)$.

	The volume of the lattice generated by $M_C$ is 
\[
	\de{I_N\otimes D_\alpha}\de{I_k\otimes M}\de{I_{N-k}\otimes M_p}=N(\alpha)^{N/2}\Delta^{\frac{k}{2}}\left(\Delta^{\frac{1}{2}}p^f\right)^{N-k},
\]
which agrees with the volume $N(\alpha)^{N/2}\Delta^{\frac{N}{2}}p^{f(N-k)}$ of $(\rho^{-1}(C),b_\alpha)$.

Define $\psi:\sigma(x)\mapsto x\in\O_K$ to be the inverse of the embedding 
\[
	\sigma=(\sqrt{\sigma_1(\alpha)}\sigma_1,\ldots,\sqrt{\sigma_n(\alpha)}\sigma_n):\O_K\hookrightarrow\RR^n.
\]
Then it suffices to prove that
\[
	\rho^{-1}(C)\supseteq\{\psi(\xv M_C):\xv\in\ZZ^{nN}\},
\]
or 
\[
	C \supseteq \{\rho(\psi(\xv M_C)):\xv\in\ZZ^{nN}\}.
\]
For $j=1,2,\dots,N$, let $\uv_j=(u_{j1},\dots,u_{jn})\in\ZZ^n$.
Then $\xv\in\ZZ^{nN}$ can be written as $\xv=(\uv_1,\dots,\uv_N)$.
Let $x_j=\sum_{i=1}^nu_{ji}v_i$, then the $s$th entry of $u_jM D_\alpha$ is given by
\[
	\sum_{i=1}^nu_{ji}\sigma_s(v_i)\sqrt{\sigma_s(\alpha)}=\sqrt{\sigma_s(\alpha)}\sigma_s\left(\sum_{i=1}^nu_{ji}v_i\right)=\sqrt{\sigma_s(\alpha)}\sigma_s(x_j).
\]
Thus
\begin{eqnarray*}
	\xv M_C&=&[\uv_1,\dots,\uv_N]\begin{bmatrix}
	I_k\otimes M& A\widetilde{\otimes}M\\
	\boldsymbol{0}_{nN-nk,nk}&I_{N-k}M_p
\end{bmatrix}(I_N\otimes D_\alpha)\\
&=&\left[\sigma(x_1),\dots,\sigma(x_k),\sigma\left(\sum_{j=1}^ka_{j1}x_j+x'_{k+1}\right),\dots,\sigma\left(\sum_{j=1}^ka_{j(N-k)}x_j+x'_{N}\right)\right],
\end{eqnarray*}
where $x'_{k+1},\dots,x'_N$ are in the ideal $\p$, then $\rho(x'_i)=0$ for $i=k+1,\dots,x_N$.
We have
\begin{eqnarray*}
	&&\rho(\psi(\xv M_C))\\&=&\rho(x_1,\dots,x_k,\sum_{j=1}^ka_{j1}x_j+x'_{k+1},\dots,\sum_{j=1}^ka_{j(N-k)}x_j+x'_N)\\
	&=&(x_1~{\rm mod}~\p,\dots,x_k~{\rm mod}~\p,\sum_{j=1}^ka_{j1}x_j+x'_{k+1}~{\rm mod}~\p,\dots,\sum_{j=1}^ka_{j(N-k)}x_j+x'_N~{\rm mod}~\p)\\
	&=&(x_1~{\rm mod}~\p,\dots,x_k~{\rm mod}~\p)(I_k~A~{\rm mod}~\p)\in C.
\end{eqnarray*}
\end{proof}
\begin{lem}
	The Gram matrix $G_C=M_CM_C^T$ of $(\rho^{-1}(C),b_\alpha)$ is
\begin{equation}\label{eq:GCReal}
G_C=
\begin{bmatrix}
	\Tr{\alpha(I_k+AA^T)\otimes M_1M_1^T}&\Tr{\alpha A\otimes M_1 M_{p,1}^T}\\
	\Tr{\alpha A\otimes M_1M_{p,1}^T}^T&\Tr{\alpha I_{N-k}\otimes M_{p,1}M_{p,1}^T}
\end{bmatrix}
\end{equation}
where ${\rm Tr}={\rm Tr}_{K/\QQ}$ is taken componentwise and $M_{p,1}$ denotes the first column of the matrix $M_p$.
\end{lem}
\begin{proof}
Let $\tilde{D_\alpha}=D_\alpha D_\alpha^T$ be the diagonal matrix with diagonal entries given by $\sigma_1(\alpha),\dots,\sigma_n(\alpha)$.
For $M_C$ in (\ref{eq:MCReal}), a direct computation gives
\[
G_C=\begin{bmatrix}
	I_k\otimes M\tilde{D}_\alpha M^T+(A\widetilde{\otimes}M)(I_{N-k}\otimes\tilde{D}_\alpha)(A\widetilde{\otimes}M)^T&(A\widetilde{\otimes}M)(I_{N-k}\otimes \tilde{D}_\alpha M_p^T)\\
	(I_{N-k}\otimes M_p\tilde{D}_\alpha)(A\widetilde{\otimes}M)^T&I_{N-k}\otimes M_p\tilde{D}_\alpha M_p^T
\end{bmatrix}.
\]
Using that $M_i=\sigma_i(M_1)$, $(1\leq i\leq n)$, we have
\begin{eqnarray*}
	(A\widetilde{\otimes} M)(I_{N-k}\otimes\tilde{D}_\alpha)(A\widetilde{\otimes} M)^T &= &\Tr{\alpha AA^T\otimes M_1M_1^T}\\
	I_k\otimes M\tilde{D}_\alpha M^T&=&\Tr{\alpha I_k\otimes M_1M_1^T}
\end{eqnarray*}
thus showing that
\begin{equation*}
I_k\otimes M\tilde{D}_\alpha M^T+(A\widetilde{\otimes}M)(I_{N-k}\otimes\tilde{D}_\alpha)(A\widetilde{\otimes}M)^T=\Tr{\alpha(I_k+AA^T)\otimes M_1M_1^T}.
\end{equation*}
Similarly, let $M_{p,i}$ denote the $i$th column of $M_p$, then using $\sigma_i(M_{p,1})=M_{p,i}$ $(1\leq i\leq n)$, we have
\[
	I_{N-k}\otimes M_p\tilde{D}_\alpha M_p^T=\Tr{\alpha I_{N-k}\otimes M_{p,1}M_{p,1}^T}.
\]
Moreover,
\begin{equation*}
  \begin{array}{l}
	  (A\widetilde{\otimes}M)(I_{N-k}\otimes \tilde{D}_\alpha M_p^T)=\\
    \begin{bmatrix}
	    \sigma_1(a_{11})M_1&\sigma_2(a_{11})M_2& \!\!\!\!\!\ldots \!\!\!\!\!& \sigma_n(a_{1,{N-k}})M_n\\
      \vdots   &  \vdots &    & \vdots\\
      \sigma_1(a_{k,1})M_1&\sigma_2(a_{k,2})M_2& \!\!\!\!\! \ldots \!\!\!\!\!    &\sigma_n(a_{k,N-k})M_n
    \end{bmatrix}
   \begin{bmatrix}
	\sigma_1(\alpha)M_{p,1}^T&0&\!\!\!\!\!\ldots \!\!\!\!\!&0\\
	\sigma_2(\alpha)M_{p,2}^T&0&\!\!\!\!\!\ldots \!\!\!\!\!&0\\
        \vdots&\vdots&\ddots&\vdots\\
	\sigma_n(\alpha)M_{p,n}^T&0&\!\!\!\!\!\ldots \!\!\!\!\!&0\\
        \vdots&\vdots&\ddots&\vdots\\
	0&0&\!\!\!\!\!\ldots \!\!\!\!\!&\sigma_{n-1}(\alpha)M_{p,n-1}^T\\
	0&0&\!\!\!\!\!\ldots \!\!\!\!\!&\sigma_n(\alpha)M_{p,n}^T
    \end{bmatrix}\\
    =\sigma_1(A\otimes \alpha M_1M_{p,1}^T)+\sigma_2(A\otimes \alpha M_1M_{p,1}^T)+\dots+\sigma_n(A\otimes\alpha M_1M_{p,1}^T)\\
    =\Tr{\alpha A\otimes M_1M_{p,1}^T},
  \end{array}
\end{equation*}
\end{proof}

When $K$ is a CM number field, $n$ is even and all embeddings of $K$ into $\CC$ are complex embeddings.
Assume $\sigma_{i+1}$ is the conjugate of $\sigma_i$ for $i=1,3,5,\dots,n-1$.

\begin{lemma}
Let $K$ be a CM number field of degree $n$.
Then
\begin{equation}\label{eqn:MCM}
	M=\sqrt{2}
\begin{bmatrix}
	\re\sigma_1(v_1)&\im\sigma_2(v_1)&\re\sigma_3(v_1)&\dots&\re\sigma_{n-1}(v_1)&\im\sigma_n(v_1)\\
	\vdots&\vdots&\vdots&\ddots&\vdots&\vdots\\
	\re\sigma_1(v_n)&\im\sigma_2(v_n)&\re\sigma_3(v_n)&\dots&\re\sigma_{n-1}(v_n)&\im\sigma_n(v_n)
\end{bmatrix}
\end{equation}
is a generator matrix for the lattice $(\O_K,b_1)$ and $\de{M}=\Delta^{\frac{1}{2}}$.
\end{lemma}
\begin{proof}
	The Gram matrix for $(\O_K,b_1)$ is $G=(\Tr{v_i\bar{v}_j})_{1\leq i,j\leq n}$.
	For $i=1,2,\dots,n$,
	\begin{eqnarray*}
		(MM^T)_{ii}&=&2\sum_{j=1,3,\dots,n-1}(\re\sigma_j(v_i))^2+(\im\sigma_{j+1}(v_i))^2=2\sum_{j=1,3,\dots,n-1}|\sigma_j(v_i)|^2\\
		&=&2\sum_{j=1,3,\dots,n-1}\sigma_j(|v_i|^2)=\Tr{|v_i|^2}=G_{ii}.
	\end{eqnarray*}
	For $i,j=1,2,\dots,n,i\neq j$,
	\begin{eqnarray*}
		(MM^T)_{ij}&=&2\sum_{s=1,3,\dots,n-1}\re\sigma_s(v_i)\re\sigma_s(v_j)+\im\sigma_{s+1}(v_i)\im\sigma_{s+1}(v_j)\\
		&=&2\sum_{s=1,3,\dots,n-1}\re\sigma_s(v_i)\re\sigma_s(v_j)+\im\sigma_s(v_i)\im\sigma_s(v_j)\\
		&=&2\sum_{s=1,3,\dots,n-1}\re\sigma_s(v_i\bar{v}_j)=\Tr{v_i\bar{v}_j}=G_{ij}.
	\end{eqnarray*}
	The determinant of $M$ is then given by the volume of $(\O_K,b_1)$.
\end{proof}
Define 
\begin{equation}\label{eqn:MpCM}
M_p=\sqrt{2}
\begin{bmatrix}	\re\sigma_1(\omega_1)&\im\sigma_2(\omega_1)&\re\sigma_3(\omega_1)&\dots&\re\sigma_{n-1}(\omega_1)&\im\sigma_n(\omega_1)\\
	\vdots&\vdots&\vdots&\ddots&\vdots&\vdots\\
	\re\sigma_1(\omega_n)&\im\sigma_2(\omega_n)&\re\sigma_3(\omega_n)&\dots&\re\sigma_{n-1}(\omega_n)&\im\sigma_n(\omega_n)
\end{bmatrix}.
\end{equation}
Then similarly $M_p$ is a generator matrix for $(\p,b_1)$ and has determinant $\Delta^{\frac{1}{2}}p^f$.
As $\alpha$ is totally positive, all $\sigma_i(\alpha)\in\RR$.
Let $D_\alpha$ be a diagonal matrix whose diagonal entries are $\sqrt{\sigma_i(\alpha)}$, $i=1,\ldots,n$.

\begin{proposition}\label{prop:MCCM}
	Let $K$ be a CM field with degree $n$ and Galois group $\{\sigma_1,\sigma_2,\dots,\sigma_n\}$, where $\sigma_{i+1}$ is the conjugate of $\sigma_i$ $(i=1,3,\dots,n-1)$.
       	A generator matrix for $(\rho^{-1}(C),b_\alpha)$ is given by
	\begin{equation}\label{eq:MCCM}
		M_C=\begin{bmatrix}
			I_k\otimes M& A\widetilde{\tilde{\otimes}} M\\
	\boldsymbol{0}_{nN-nk,nk}&I_{N-k}\otimes M_p
\end{bmatrix}(I_N\otimes D_\alpha),
	\end{equation}
	where $M$ and $M_p$ are defined in (\ref{eqn:MCM}) and (\ref{eqn:MpCM}) respectively.
	$A$ is a matrix such that $(I_k~(A~{\rm mod}~\p))$ is a generator matrix of $C$ and
	\begin{equation*}
		\begin{array}{l}
		A\widetilde{\tilde{\otimes}}M:=[
		{\rm Re}\sigma_1(A_1) \otimes M_1+{\rm Im}\sigma_1(A_1)\otimes M_2 , {\rm Re}\sigma_1(A_1) \otimes M_2-{\rm Im}\sigma_1(A_1)\otimes M_1,\\
		\re\sigma_3(A_1)\otimes M_3+\im\sigma_3(A_1)\otimes M_4, \re\sigma_3(A_1)\otimes M_4-\im\sigma_3(A_1)\otimes M_3,\dots,\\
		{\rm Re}\sigma_{n-1}(A_1) \otimes M_{n-1}+{\rm Im}\sigma_{n-1}(A_1)\otimes M_n , {\rm Re}\sigma_{n-1}(A_1) \otimes M_n-{\rm Im}\sigma_{n-1}(A_1)\otimes M_{n-1},\\
		\dots,{\rm Re}\sigma_{n-1}(A_{N-k}) \otimes M_{n-1}+{\rm Im}\sigma_n(A_{N-k})\otimes M_n,\\
 {\rm Re}\sigma_{n-1}(A_{N-k}) \otimes M_n-{\rm Im}\sigma_{n-1}(A_{N-k})\otimes M_{n-1}],
		 \end{array}
	\end{equation*}
	where we denote the columns of the matrices $M,A$ by $M_i$, $i=1,\ldots,n$, $A_j$, $j=1,2,\dots,N-k$, Re and Im are understood componentwise.
\end{proposition}

\begin{proof}
	The volume of the lattice generated by $M_C$ is
\begin{eqnarray*}
	\de{I_k\otimes M}\de{I_{N-k}\otimes M_p}\de{I_N\otimes D_\alpha}
&=&\Delta^{\frac{k}{2}}\left(\Delta^{\frac{1}{2}}p^f\right)^{N-k}N(\alpha)^{N/2}\\
&=&\Delta^{\frac{N}{2}}p^{f(N-k)}N(\alpha)^{N/2},
\end{eqnarray*}
which agrees with the volume of $(\rho^{-1}(C),b_\alpha)$.

Define $\psi:\sigma(x)\mapsto x\in\O_K$ to be the inverse of the embedding 
\[
	\sigma=\sqrt{2}(\sqrt{\sigma_1(\alpha)}\re\sigma_1,\sqrt{\sigma_2(\alpha)}\im\sigma_2,\dots,\sqrt{\sigma_{n-1}(\alpha)}\re\sigma_{n-1},\sqrt{\sigma_n(\alpha)}\im\sigma_n):\O_K\hookrightarrow\RR^n.
\]
Then it suffices to prove
\[
	\rho^{-1}(C)\supseteq\{\psi(\xv M_C):\xv\in\ZZ^{nN}\},
\]
or
\[
	\{\rho(\psi(\xv M_C)):\xv\in\ZZ^{nN}\}\subseteq C.
\]
For $j=1,2,\dots,N$, let $\uv_j=(u_{j1},\dots,u_{jn})\in\ZZ^n$.
Then $\xv\in\ZZ^{nN}$ can be written as $\xv=(\uv_1,\dots,\uv_N)$.
Let $x_j=\sum_{i=1}^nu_{ji}v_i$, we have the $t$th entry of $\uv_j MD_\alpha$ is
\[
	\sqrt{2}\sum_{i=1}^nu_{ji}\re\sigma_t(v_i)\sqrt{\sigma_t(\alpha)}=\sqrt{2}\sqrt{\sigma_t(\alpha)}\sum_{i=1}^n\re\sigma_t(u_{ji}v_i)=\sqrt{2}\sqrt{\sigma_t(\alpha)}\re\sigma_t(x_j)
\]
for $t$ odd, or $\sqrt{2}\sqrt{\sigma_t(\alpha)}\im\sigma_t(x_j)$ for $t$ even.
And the $st$th entry $(1\leq s\leq N-k,1\leq t\leq n)$ of $\uv_j
(A\widetilde{\tilde{\otimes}} M)(I_{N-k}\otimes D_\alpha)$ $(1\leq j\leq k)$ is
\begin{eqnarray*}
	&\ &\sqrt{2}\sum_{i=1}^nu_{ji}\left[\re\sigma_t(a_{js})\re\sigma_t(v_i)+\im\sigma_t(a_{js})\im\sigma_{t+1}(v_i)\right]\sqrt{\sigma_t(\alpha)}\\
	&=&\sqrt{2}\sqrt{\sigma_t(\alpha)}\sum_{i=1}^n\re\sigma_t(a_{js})\re\sigma_t(u_{ji}v_i)+\im\sigma_t(a_{js})\im\sigma_{t+1}(u_{ji}v_i)\\
	&=&\sqrt{2}\sqrt{\sigma_t(\alpha)}\sum_{i=1}^n\re\sigma_t(a_{js}u_{ji}v_i)=\sqrt{2}\sqrt{\sigma_t(\alpha)}\re\sigma_t(a_{js})x_j
\end{eqnarray*}
for $t$ odd, or $\sqrt{2}\sqrt{\sigma_t(\alpha)}\im\sigma_t(a_{js})x_j$ for $t$ even.

Then
\begin{eqnarray*}
	\xv M_C&=&[\uv_1,\dots,\uv_N]\begin{bmatrix}
		I_k\otimes M& A\widetilde{\tilde{\otimes}} M\\
	\boldsymbol{0}_{nN-nk,nk}&I_{N-k}\otimes M_p
\end{bmatrix}(I_N\otimes D_\alpha)\\
&=&\left[\sigma(x_1),\dots,\sigma(x_k),\sigma\left(\sum_{j=1}^ka_{j1}x_j+x'_{k+1}\right),\dots,\sigma\left(\sum_{j=1}^ka_{j(N-k)}x_j+x'_{N}\right)\right],
\end{eqnarray*}
where $x'_{k+1},\dots,x'_N$ are in the ideal $\p$, and hence $\rho(x'_i)=0$ for $i=k+1,\dots,x_N$.
Then we have
\begin{eqnarray*}
	&&\rho(\psi(\xv M_C))\\&=&\rho(x_1,\dots,x_k,\sum_{j=1}^ka_{j1}x_j+x'_{k+1},\dots,\sum_{j=1}^ka_{j(N-k)}x_j+x'_N)\\
	&=&(x_1~{\rm mod}~\p,\dots,x_k~{\rm mod}~\p,\sum_{j=1}^ka_{j1}x_j+x'_{k+1}~{\rm mod}~\p,\dots,\sum_{j=1}^ka_{j(N-k)}x_j+x'_N~{\rm mod}~\p)\\
	&=&(x_1~{\rm mod}~\p,\dots,x_k~{\rm mod}~\p)(I_k~A~{\rm mod}~\p)\in C.
\end{eqnarray*}
\end{proof}
\begin{remark}\label{rm:AOMCM}
\item 1. Let $\vv=[v_1,v_2,\dots,v_n]^T$, then 
	\begin{eqnarray*}
		A\widetilde{\tilde{\otimes}}M&=&\sqrt{2}[\re\sigma_1(A_1\otimes\vv),\im\sigma_2(A_1\otimes\vv),\\
&&\dots,\re\sigma_{n-1}(A_1\otimes\vv),\im\sigma_n(A_1\otimes\vv),\dots,\im\sigma_n(A_{N-k}\otimes\vv)].
	\end{eqnarray*}
\item 2. When $p$ is totally ramified, the entries of $A~{\rm mod}~\p$ are in $\FF_p$ and hence $A\widetilde{\tilde{\otimes}}M=A\otimes M$.
\end{remark}
\begin{proposition}\label{prop:gramIm}
	The Gram matrix $G_C=M_CM_C^T$ of $(\rho^{-1}(C),b_\alpha)$ is
\begin{equation}\label{eq:GCCM}
G_C=
\begin{bmatrix}
	\Tr{\alpha(I+AA^\dagger)\otimes \vv\vv^\dagger}&\Tr{\alpha A\otimes (\vv\wv^\dagger)}\\
	\Tr{\alpha A^T\otimes(\bar{\wv}\vv^T)}&\Tr{\alpha I_{N-k}\otimes\wv\wv^\dagger}
\end{bmatrix}
\end{equation}
where ${\rm Tr} = {\rm Tr}_{K/\QQ}$ is taken componentwise and $\wv=[w_1,w_2,\dots,w_n]^T$.
$\vv^\dagger=\bar{v}^T$, is the conjugate transpose of $\vv$.
Similarly $A^\dagger=\bar{A}^T,\ \wv^\dagger=\bar{\wv}^T$.
\end{proposition}
\begin{proof}
Let $\tilde{D_\alpha}=D_\alpha D_\alpha^T$ be the diagonal matrix with diagonal entries given by $\sigma_1(\alpha),\dots,\sigma_n(\alpha)$.
For $M_C$ in (\ref{eq:MCCM}), a direct computation gives
\[
G_C=\begin{bmatrix}
	I_k\otimes M\tilde{D}_\alpha M^T+(A\widetilde{\tilde{\otimes}}M)(I_{N-k}\otimes\tilde{D}_\alpha)(A\widetilde{\tilde{\otimes}}M)^T&(A\widetilde{\tilde{\otimes}}M)(I_{N-k}\otimes \tilde{D}_\alpha M_p^T)\\
	(I_{N-k}\otimes M_p\tilde{D}_\alpha)(A\widetilde{\tilde{\otimes}}M)^T&I_{N-k}\otimes M_p\tilde{D}_\alpha M_p^T
\end{bmatrix}.
\]
By Remark \ref{rm:AOMCM},
\begin{eqnarray*}
	&&(A\widetilde{\tilde{\otimes}}M)(I_{N-k}\otimes\tilde{D}_\alpha)(A\widetilde{\tilde{\otimes}}M)^T \\&=& \Tr{(\alpha A_1\otimes\vv)(A_1^\dagger\otimes\vv^\dagger)}+\dots+\Tr{(\alpha A_{N-k}\otimes\vv)(A_{N-k}^\dagger\otimes\vv^\dagger)}\\
&=&\Tr{\alpha A_1A_1^\dagger\otimes\vv\vv^\dagger}+\dots+\Tr{\alpha A_{N-k}A_{N-k}^\dagger\otimes\vv\vv^\dagger}\\
&=&\Tr{\alpha(A_1A_1^\dagger+\dots+A_{N-k}A_{N-k}^\dagger)\otimes\vv\vv^\dagger}\\
&=&\Tr{\alpha AA^\dagger\otimes\vv\vv^\dagger}.
\end{eqnarray*}
Furthermore,
\[
	I_k\otimes M\tilde{D}_\alpha M^T=\Tr{\alpha I_k\otimes\vv\vv^\dagger},\text{ and }I_{N-k}\otimes M_p\tilde{D}_\alpha M_p^T=\Tr{\alpha I_{N-k}\otimes\wv\wv^\dagger}.
\]
Hence
\[
	I_k\otimes M\tilde{D}_\alpha M^T+(A\widetilde{\tilde{\otimes}}M)(I_{N-k}\otimes\tilde{D}_\alpha)(A\widetilde{\tilde{\otimes}}M)^T=\Tr{\alpha(I_k+AA^\dagger)\otimes \vv\vv^\dagger}.
\]
Next, it can be computed that 
\[
(I_{N-k}\otimes M_p\tilde{D}_\alpha)(A\widetilde{\tilde{\otimes}}M)^T=
\Tr{\alpha A^T\otimes(\bar{\wv}\vv^T)}
\]
which also gives
\[
	(A\widetilde{\tilde{\otimes}}M)(I_{N-k}\otimes \tilde{D}_\alpha M_p^T)=\Tr{\alpha A\otimes (\vv\wv^\dagger)}.
\]
So the Gram Matrix is given by (\ref{eq:GCCM}).
\end{proof}
In Sections \ref{sec:real} and \ref{sec:imaginary}, we will consider particular cases when $\alpha=1/p$ or $1/2p$ for $K$ a real quadratic field with $\p$ inert and $K$ an imaginary quadratic field with $\p$ totally ramified.
As we are interested in constructions of modular lattices, which are integral lattices, the following proposition justifies why we will focus on self-orthogonal codes in the future.
\begin{proposition}\label{prop:selforthint}
	If $C$ is not self-orthogonal, i.e. if $C\nsubseteq C^\perp$, then $(\rho^{-1}(C),b_\alpha)$ is not an integral lattice for any $\alpha\in\p^{-1}\cap\QQ$ when
\item 1. $K$ is totally real, or
\item 2. $K$ is a CM field and $\p$ is totally ramified.
\end{proposition}
\begin{proof}
	Let $\{\tilde{c_1},\tilde{c_2},\dots,\tilde{c_k}\}$ be an $\FF_{p^f}-$basis for the linear code $C$.
	Let $\{\boldsymbol{c}_1,\boldsymbol{c}_2,\dots,\boldsymbol{c}_k\}$ be a set of elements in $\O_K^N$ such that $\boldsymbol{c}_i$ is a preimage of $\tilde{c}_i$, $1\leq i\leq k$.
	Notice that for $1\leq i\leq k,\ 1\leq j\leq n$,
	\[
		\rho(v_j\boldsymbol{c}_i)=\rho(v_j)\rho(\boldsymbol{c}_i)=\rho(v_j)\tilde{c_i}.
	\]
	As $\rho(v_j)\in\FF_{p^f}$, $\rho(v_j\boldsymbol{c}_i)\in C$, i.e., $v_j\boldsymbol{c}_i\in\rho^{-1}(C)$ for all $1\leq i\leq k$ and $1\leq j\leq n$. 
	Since $C\nsubseteq C^\perp$, take $\boldsymbol{c}_{i_1},\boldsymbol{c}_{i_2}$ such that $\tilde{c}_{i_1}\cdot\tilde{c}_{i_2}\neq0$, i.e., $\boldsymbol{c}_{i_1}\cdot\boldsymbol{c}_{i_2}\notin\p$.
	From $\{v_1,\dots,v_2\}$, take $v_{j_0}$ such that $v_{j_0}\notin\p$.
	Suppose
	\[
		b_\alpha(v_j\boldsymbol{c}_{i_1},v_{j_0}\boldsymbol{c}_{i_2})=\alpha\Tr{v_{j_0}\boldsymbol{c}_{i_1}\cdot\bar{\boldsymbol{c}}_{i_2}\bar{v}_j}\in\ZZ
	\]
for all $1\leq j\leq n$.
As $\{v_j\}_{1\leq j\leq n}$ forms a basis for $\O_K$, $\{\bar{v}_j\}_{1\leq j\leq n}$ also forms a basis for $\Oc_K$, we have
\[
	\alpha\Tr{v_{j_0}\boldsymbol{c}_{i_1}\cdot\bar{\boldsymbol{c}}_{i_2}x}\in\ZZ\ \forall x\in\O_K.
\]
By the definition of the codifferent $\D_K^{-1}$,
\[
	\alpha v_{j_0}\boldsymbol{c}_{i_1}\cdot\bar{\boldsymbol{c}}_{i_2}\in\D_K^{-1}\Longrightarrow v_{j_0}\boldsymbol{c}_{i_1}\cdot\bar{\boldsymbol{c}}_{i_2}\in \alpha^{-1}\D_K^{-1}\cap\O_K=\alpha^{-1}\Oc_K\subseteq\p.
\]
As $v_{j_0}\notin\p$, we have $\boldsymbol{c}_{i_1}\cdot\bar{\boldsymbol{c}}_{i_2}\in\p$.
For $K$ totally real, this is the same as $\boldsymbol{c}_{i_1}\cdot\boldsymbol{c}_{i_2}\in\p$.
For $K$ CM, as $\p$ is totally ramified, by the proof from \cite{Ong}, $\beta\equiv \bar{\beta}~{\rm mod}~\p$ for all $\beta\in\O_K$.
It goes as follows:

As $\O_K/\p\cong\FF_p$, we can write $\beta=\beta'+\beta''$ with $\beta'\in\ZZ$ and $\beta''\in\p$.
Since $\p$ is the only prime above $p$, $\bar{\p}=\p$ and we have $\bar{\beta''}\in\p$.
Thus
	\[
		\bar{\beta}=\bar{\beta'}+\bar{\beta''}=\beta'+\bar{\beta''}\equiv\beta'~{\rm mod}~\p\equiv\beta~{\rm mod}~\p.
	\]
Then we can conclude $\boldsymbol{c}_{i_1}\cdot\boldsymbol{c}_{i_2}\in\p$.
For both cases, we get a contradiction with the choice of $\boldsymbol{c}_{i_1}$ and $\boldsymbol{c}_{i_2}$.

Thus we must have $b_\alpha(v_j\boldsymbol{c}_{i_1},v_{j_0}\boldsymbol{c}_{i_2})\notin\ZZ$ for at least one $j$ $(1\leq j\leq n)$.
As $v_j\boldsymbol{c}_{i_1},v_{j_0}\boldsymbol{c}_{i_2}\in\rho^{-1}(C)$ for all $j$, we can conclude that the lattice $(\rho^{-1}(C),b_\alpha)$ is not integral.
\end{proof}

%
%
%
\section{Modular Lattices from Totally Real Quadratic Fields}\label{sec:real}

Let $d$ be a positive square-free integer.  
Let $K=\QQ(\sqrt{d})$ be a totally real quadratic field with Galois group $\{\sigma_1,\sigma_2\}$
and discriminant $\Delta$ given by \cite{Neukirch}:
\[
\Delta=\left\{
\begin{array}{ll}
  d&d\equiv1~{\rm mod}~4\\
  4d&d\equiv2,3~{\rm mod}~4
\end{array}
\right..
\]

Assume $p\in\ZZ$ is a prime which is inert in $K$, and consider the lattice $(\rho^{-1}(C),b_\alpha)$ 
where $C$ is a linear $(N,k)$ code over $\FF_{p^2}$.

Let $\alpha=1/p$ when $d\equiv1~{\rm mod}~4$ and let $\alpha=1/2p$ when $d\equiv2,3~{\rm mod}~4$.
We will give two proofs that if $C$ is self-dual (i.e., $C=C^\perp$), then the lattice $(\rho^{-1}(C),b_\alpha)$ is a $d$-modular lattice.
Note that the results in \cite{ITW2014} are corollaries from results in this section by taking $d=5$.

By the discussion from Section \ref{sec:matrix}, a generator matrix for $(\rho^{-1}(C),b_\alpha)$ is (see (\ref{eq:MCReal}))
\begin{equation}\label{eq:MC}
	M_C=\sqrt{\alpha}
\left[
\begin{array}{cc}
I_k\otimes M & A \widetilde{\otimes} M \\
{\bf 0}_{2N-2k,2k}& I_{N-k}\otimes pM
\end{array}
\right]
\end{equation}
where $(I_k, (A \mod p\Oc_K))$ is a generator matrix for $C$, 
\begin{equation}\label{eq:M}
M=
\begin{bmatrix}
  1&1\\
  \sigma_1(v)&\sigma_2(v)
\end{bmatrix},
\end{equation}
with $\{1,v\}$ a $\ZZ-$basis of $\mathcal{O}_K$, and
\[
v=\left\{
	\begin{array}{ll}
		\frac{1+\sqrt{d}}{2}&d\equiv1~{\rm mod}~4\\
		\sqrt{d}&d\equiv2,3~{\rm mod}~4
	\end{array}
\right..
\]
Also, the Gram matrix for $(\rho^{-1}(C),b_{\alpha})$ is given by (see (\ref{eq:GCReal}))
\begin{equation}\label{eq:GC}
G_C=\alpha
\begin{bmatrix}
  \Tr{(I+AA^T)\otimes M_1M_1^T}&p\Tr{A\otimes M_1M_1^T}\\
  p\Tr{A\otimes M_1M_1^T}^T&I_k\otimes p^2 MM^T
\end{bmatrix}.
\end{equation}
Note that since $p$ is inert, $M_p=pM$.

\begin{lem}
If $C$ is self-orthogonal, then the lattice $(\rho^{-1}(C),b_\alpha)$ with $\alpha=1/p$ when $d\equiv1~{\rm mod}~4$ and $\alpha=1/2p$ when $d\equiv2,3~{\rm mod}~4$ is integral.
\end{lem}
\begin{proof}
An equivalent definition of integral lattice is that its Gram matrix has integral coefficients,
which is the case: $MM^T$ has integral coefficients, both $A$ and $I+AA^T$ have coefficients in $\O_K$, thus $\Tr{(I+AA^T)\otimes M_1M_1^T}$ and $\Tr{A\otimes M_1M_1^T}$ have integral coefficients.

As $C$ is self-orthogonal and $(I_k~~A~{\rm mod}~ (p))$ is a generator matrix for $C$, $I_k+AA^T\equiv0~{\rm mod}~ (p)$. 
Hence $(I_k+AA^T)\otimes M_1M_1^T\in (p)$ (see Lemma \ref{lem:selfdualcode}) and $\Tr{(I_k+AA^T)\otimes M_1M_1^T}\in p\ZZ$.

Finally, for the case $d\equiv2,3~{\rm mod}~ 4$, any entry in $\alpha^{-1}G_C$ is an element of $\O_K$. Since $\O_K=\ZZ[\sqrt{d}]$, for any $x=a+b\sqrt{d}\in\O_K$, $\Tr{x}=2a\in2\ZZ$.
\end{proof}

We can tell the duality properties of a linear code from its generator matrix \cite{Ling}:
\begin{lemma}\label{lem:selfdualcode}
	Let $C$ be a linear code over $\FF_q$, let $B$ be a generator matrix for $C$. 
	A matrix $H\in M_{(N-k)\times N}(\FF_q)$ is a parity check matrix for $C$ iff $HB^T={\bf 0}$. 
	In particular,
	
	1. if $B=(I_k~A)$, then $(-A^T~I_{N-k})$ is a parity check matrix for $C$;

	2. $C$ is self-dual iff $I+AA^T={\bf 0}$.
\end{lemma}
Hence if $C$ is self-dual and $(I_k~~(A~{\rm mod}~p\O_K))$ is a generator matrix of $C$, then $((-A^T~{\rm mod}~p\O_K)~~I_k)$ is also a generator matrix of $C$ and $N-k=k$.

We propose next two approaches to discuss the modularity of lattices obtained via the above method.
\subsection{Approach I} We will use the knowledge of a generator matrix of the lattice.
\begin{remark}\label{rem:gendualgen}
	Note that \cite{ITW2014}

\item	1. If $M$ is a generator matrix for $(L,b)$, then $M^*:=(M^T)^{-1}$ is a generator matrix for $(L^*,b)$.
\item	2. $(L,b)$ is $d-$modular if and only if $\frac{1}{\sqrt{d}}M$ is a generator matrix for $(L^*,b)$.

	Here $b$ denotes any positive symmetric bilinear form.
\end{remark}

We get another generator matrix for $(\rho^{-1}(C),b_\alpha)$:
\begin{proposition}\label{MC2}
	If $C$ is self-dual, another generator matrix of $(\rho^{-1}(C),b_\alpha)$ is
\begin{equation}\label{eq:MC2}
	M'_C=\sqrt{\alpha}
\begin{bmatrix}
-A^T \widetilde{\otimes} M & I_k\otimes M\\
I_{k}\otimes pM & {\bf 0}_{2k,2k}
\end{bmatrix}
\end{equation}
with $M$ as in (\ref{eq:M}), $A$ such that $(I_k~~(A~{\rm mod}~p\O_K))$ is a generator matrix of $C$.
\end{proposition}
\begin{proof}
  Let $b_{ij}$ denote the entries of $-A^T$.
  Keep the same notations as in the proof of Proposition \ref{prop:MCReal}.
  Define $\psi:\sigma(x)\mapsto x\in\Oc_K$ to be the inverse of the embedding $\sigma=(\sqrt{\sigma_1(\alpha)}\sigma_1,\sqrt{\sigma_2(\alpha)}\sigma_2):\Oc_K\hookrightarrow\RR^2$.
  For $j=1,2,\dots,N$, let $\uv_j=(u_{j1},u_{j2})\in\ZZ^2$.
  Then $\xv\in\ZZ^{2N}$ can be written as $\xv=(\uv_1,\dots,\uv_N)$.
  Let $x_j=u_{j1}+u_{j2}v$ for $1\leq j\leq N$.
  Using the formula for Schur complement, we can check that this matrix has the right determinant.
  We are left to show that lattice points are indeed mapped to codewords in $C$ by $\rho$, i.e.
  \[
	  \{\rho(\psi(\xv M_C)):\xv\in\ZZ^{2N}\}\subseteq C,
  \]
  By a similar argument as in Proposition \ref{prop:MCReal}, we have 
\begin{eqnarray*}  
\xv M_C&=&[u_1,\dots,u_k,\dots,u_N]\sqrt{\alpha}
\left[
\begin{array}{cc}
 -A^T \widetilde{\otimes} M & I_k\otimes M\\
I_{k}\otimes pM & {\bf 0}_{2k,2k}
\end{array}
\right]\\
&=&[\sigma(\sum_{j=1}^kb_{j1}x_j+x'_{1}),\dots,
\sigma(\sum_{j=1}^kb_{jk}x_j+x'_k), \sigma(x_1),
\dots,\sigma(x_k)],
\end{eqnarray*}
where $x'_{1},\dots,x'_k$ are in the ideal $(p)$.
Since $x'_i$ reduces to zero mod $(p)$, we
have
\[
\psi(\xv)=
\psi(\sigma(\sum_{j=1}^k b_{j1}x_j+x'_{1})),\dots,(\psi(\sigma(x_1)),\ldots),
\]
and $\rho\psi(\xv)$ is indeed a codeword of $C$:
\begin{eqnarray*}
\rho\psi(\xv)
&=& (\sum_{j=1}^k b_{j1}x_j+x'_{1}~{\rm mod }(p),\dots,x_1~{\rm mod }(p),\ldots)\\
&=& (x_1~{\rm mod }(p),\dots, x_k~{\rm mod }(p))\cdot(-A^T~{\rm mod }(p)~~I_k)
\end{eqnarray*}
\end{proof}
To continue, we need the following lemma, which can be proved by direct computation (see Remark \ref{rem:gendualgen}):
\begin{lem}\label{lem:relationsMC}
	1. For $d\equiv1~{\rm mod}~4$, $(\O_K,b_1)$ is $d-$modular, i.e. $\frac{1}{\sqrt{d}}M=UM^*$ for some integral matrix $U$ with determinant $\pm1$.

	2. For $d\equiv2,3~{\rm mod}~4$, define
\[
M_{\mathfrak{P}_2^{-1}}=\frac{1}{\sqrt{2}}M
=
\frac{1}{\sqrt{2}}
\begin{bmatrix}
  1&1\\
  \sqrt{d}&-\sqrt{d}
\end{bmatrix}.
\]
Then $M_{\P_2^{-1}}$ is a generator matrix for $(\O_K,\frac{1}{2}b_1)$ and $(\O_K,\frac{1}{2}b_1)$ is $d-$modular, i.e. $\frac{1}{\sqrt{d}}M_{\mathfrak{P}_2^{-1}}=UM_{\mathfrak{P}_2^{-1}}^*$ for some integral matrix $U$ with determinant $\pm1$.
\end{lem}

\begin{proposition}
Let $C$ be a self-dual code.
The lattice $(\rho^{-1}(C),b_\alpha)$ is $d$-modular.
\end{proposition}
\begin{proof}
	\textbf{Case 1:} $d\equiv1~{\rm mod}~4$. By Remark \ref{rem:gendualgen}, a generator matrix for the dual of $\rho^{-1}(C)$ with respect to the bilinear form $(\xv,\yv)\mapsto \frac{1}{p}\sum_{i=1}^N\Tr{x_iy_i}$ is $(M_C^T)^{-1}$, where $M_C$ is given in (\ref{eq:MC}).
This can be computed using Schur complement:
\begin{eqnarray*}
&&\sqrt{p}
\begin{bmatrix}
  I_k\otimes M^*&\bf{0}\\
  -\frac{1}{p}(I_k\otimes M^*)(A\widetilde{\otimes}M)^T(I_k\otimes M^*)&\frac{1}{p}I_k\otimes M^*
\end{bmatrix}\\
&=&\frac{1}{\sqrt{p}}
\begin{bmatrix}
  I_k\otimes pM^*&\bf{0}\\
  -(I_k\otimes M^*)(A\widetilde{\otimes}M)^T(I_k\otimes M^*)&I_k\otimes M^*
\end{bmatrix},
\end{eqnarray*}
By a change of basis, we get another generator matrix for the dual as
\[
\frac{1}{\sqrt{p}}
\begin{bmatrix}
 -(I_k\otimes M^*)(A\widetilde{\otimes}M)^T(I_k\otimes M^*)&I_k\otimes M^*\\
  I_k\otimes pM^*&\bf{0}
\end{bmatrix}.
\]
By Lemma \ref{lem:relationsMC}, we get the following generator matrix (note that $I\otimes(UM^*)=(I\otimes U)(I\otimes M^*)$)
\[
\frac{1}{\sqrt{d p}}
\begin{bmatrix}
 -\sqrt{d}(I_k\otimes M^*)(A\widetilde{\otimes}M)^T(I_k\otimes M^*)&I_k\otimes M\\
  I_k\otimes pM&\bf{0}
\end{bmatrix}.
\]
By Proposition \ref{MC2},
\[
\frac{1}{\sqrt{p}}
\begin{bmatrix}
-A^T\widetilde{\otimes}M&I_k\otimes M\\
I_k\otimes pM&\bf{0}
\end{bmatrix}
\]
can be seen to be another generator matrix, it suffices now to prove $\sqrt{d}(I_k\otimes M^*)(A\widetilde{\otimes}M)^T(I_k\otimes M^*)=A^T\widetilde{\otimes}M$, which is equivalent to
  \begin{eqnarray*}
    &(I_k\otimes M)(A\widetilde{\otimes}M)^T(I_k\otimes M^*)=A^T\widetilde{\otimes}M\\
    \Longleftrightarrow&(I_k\otimes M)(A\widetilde{\otimes}M)^T=(A^T\widetilde{\otimes}M)(I_k\otimes M^T)\\
    \Longleftrightarrow&(A\widetilde{\otimes}M)(I_k\otimes M^T)=(I_k\otimes M)(A^T\widetilde{\otimes} M)^T\\
    \Longleftrightarrow&\Tr{A\otimes M_1M_1^T}=(I_k\otimes M)(A^T\widetilde{\otimes} M)^T
  \end{eqnarray*}
which can be checked by direct computations.\\
\textbf{Case 1:} $d\equiv2,3~{\rm mod}~4$. Similarly, a generator matrix for the dual of $\rho^{-1}(C)$ with respect to the bilinear form $(\xv,\yv)\mapsto \frac{1}{2p}\sum_{i=1}^N\Tr{x_iy_i}$ is $(M_C^T)^{-1}$. Using Schur complement:
\begin{eqnarray*}\sqrt{p}
\begin{bmatrix}
  I_k\otimes M_{\mathfrak{P}_2^{-1}}^*&\bf{0}\\
  -\frac{1}{\sqrt{2}p}(I_k\otimes M_{\mathfrak{P}_2^{-1}}^*)(A\widetilde{\otimes}M)^T(I_k\otimes M_{\mathfrak{P}_2^{-1}}^*)&\frac{1}{p}I_k\otimes M_{\mathfrak{P}_2^{-1}}^*
\end{bmatrix}\\
=\frac{1}{\sqrt{p}}
\begin{bmatrix}
  I_k\otimes pM_{\mathfrak{P}_2^{-1}}^*&\bf{0}\\
  -\frac{1}{\sqrt{2}}(I_k\otimes M_{\mathfrak{P}_2^{-1}}^*)(A\widetilde{\otimes}M)^T(I_k\otimes M_{\mathfrak{P}_2^{-1}}^*)&I_k\otimes M_{\mathfrak{P}_2^{-1}}^*
\end{bmatrix}
\end{eqnarray*}
By a change of basis and Lemma \ref{lem:relationsMC}, we get another generator matrix for the dual as
\begin{eqnarray*}
&&\frac{1}{\sqrt{p}}
\begin{bmatrix}
 -\frac{1}{\sqrt{2}}(I_k\otimes M_{\mathfrak{P}_2^{-1}}^*)(A\widetilde{\otimes}M)^T(I_k\otimes M_{\mathfrak{P}_2^{-1}}^*)&I_k\otimes M_{\mathfrak{P}_2^{-1}}^*\\
  I_k\otimes pM_{\mathfrak{P}_2^{-1}}^*&\bf{0}
\end{bmatrix}\\
&=&\frac{1}{\sqrt{dp}}
\begin{bmatrix}
 -\sqrt{\frac{d}{2}}(I_k\otimes M^*_{\mathfrak{P}_2^{-1}})(A\widetilde{\otimes}M)^T(I_k\otimes M^*_{\mathfrak{P}_2^{-1}})&I_k\otimes M_{\mathfrak{P}_2^{-1}}\\
  I_k\otimes pM_{\mathfrak{P}_2^{-1}}&\bf{0}
\end{bmatrix}
\end{eqnarray*}
By Proposition \ref{MC2},
\[
\frac{1}{\sqrt{2p}}
\begin{bmatrix}
-A^T\widetilde{\otimes}M&I_k\otimes M\\
I_k\otimes pM&\bf{0}
\end{bmatrix}
\]
can be seen to be another generator matrix, it suffices now to prove
$\sqrt{d}(I_k\otimes M_{\mathfrak{P}_2^{-1}}^*)(A\widetilde{\otimes}M)^T(I_k\otimes M_{\mathfrak{P}_2^{-1}}^*)=A^T\widetilde{\otimes}M$, which is equivalent to
  \begin{eqnarray*}
    &(I_k\otimes M_{\mathfrak{P}_2^{-1}})(A\widetilde{\otimes}M)^T(I_k\otimes M_{\mathfrak{P}_2^{-1}}^*)=A^T\widetilde{\otimes}M\\
    \Longleftrightarrow&(I_k\otimes M_{\mathfrak{P}_2^{-1}})(A\widetilde{\otimes}M)^T=(A^T\widetilde{\otimes}M)(I_k\otimes M_{\mathfrak{P}_2^{-1}}^T)\\
    \Longleftrightarrow&(A\widetilde{\otimes}M)(I_k\otimes M_{\mathfrak{P}_2^{-1}}^T)=(I_k\otimes M_{\mathfrak{P}_2^{-1}})(A^T\widetilde{\otimes} M)^T\\
    \Longleftrightarrow&(A\widetilde{\otimes}M)(I_k\otimes M^T)=(I_k\otimes M)(A^T\widetilde{\otimes} M)^T\\
    \Longleftrightarrow&\Tr{A\otimes M_1M_1^T}=(I_k\otimes M)(A^T\widetilde{\otimes} M)^T,
  \end{eqnarray*}
  which can be checked by direct computations.
\end{proof}
\subsection{Approach II}\label{sec:realalgebraic}
In this subsection, let $C\subseteq\FF_{p^2}^N$ be a linear code not necessarily having a generator matrix in the standard form.
We consider the lattice $(\rho^{-1}(C),b_\alpha)$, where $\alpha=1/p$ if $d\equiv1~{\rm mod}~4$ and $\alpha=1/2p$ if $d\equiv2,3~{\rm mod}~4$.
Thus $b_\alpha$ is the following bilinear form (see (\ref{eqn:bilinearform})):
\[
b_\alpha(\xv,\yv)=
\left\{
\begin{array}{ll}
 \frac{1}{p}\sum_{i=1}^N\Tr{x_iy_i}&d\equiv1~{\rm mod}~4\\
\frac{1}{2p}\sum_{i=1}^N\Tr{x_iy_i}&d\equiv2,3~{\rm mod}~4
\end{array}
\right..
\]
Then the dual of $\rho^{-1}(C)$ is given by $(\rho^{-1}(C)^*,b_\alpha)$, where $\rho^{-1}(C)^*:=\{\xv\in K^N:b_\alpha(\xv,\yv)\in\ZZ~\forall \yv\in\rho^{-1}(C)\}$.
We have the following relation between the dual of $\rho^{-1}(C)$ and the lattice constructed from the dual of $C$:
\begin{lem}\label{lem:CdualinGCdual}
	$\rho^{-1}(C^\perp)\subseteq\rho^{-1}(C)^*$.
\end{lem}
\begin{proof}
	Take any $\xv\in\rho^{-1}(C^\perp)$ and $\yv\in\rho^{-1}(C)$, we have
	\begin{eqnarray*}
    \rho(\xv\cdot \yv)&=&\rho\left(\sum_{i=1}^Nx_iy_i\right)=\sum_{i=1}^N\rho(x_i)\rho(y_i)\\
    &=&\rho(\xv)\cdot\rho(\yv)=0\in\FF_{p^2},
  \end{eqnarray*}
  where the last equality follows from the definition of $C^\perp$ (see (\ref{eqn:dualcode})).
  Then
  \[
  \sum_{i=1}^Nx_iy_i=\xv\cdot \yv\equiv0~{\rm mod}~(p).
  \]
  Since $p$ is inert, $\sigma_2\left(\sum_{i=1}^Nx_iy_i\right)\in(p)$, we have
  \[
	  \Tr{\sum_{i=1}^Nx_iy_i}\in(p)\cap\ZZ=p\ZZ.
  \]
  In the case $d\equiv2,3~{\rm mod}~4$, any element in $\O_K$ has even trace.
  In conclusion, we have $b_\alpha(\xv,\yv)\in\ZZ$ and hence $\rho^{-1}(C^\perp)\subseteq\rho^{-1}(C)^*$ by definition.
\end{proof}
\begin{corollary}
Let $C$ be a self-orthogonal linear code, then $\rho^{-1}(C)$ is integral.
\end{corollary}
\begin{proof}
As $C$ is self-orthogonal, we have $C\subseteq C^\perp$.
Hence by Lemma \ref{lem:CdualinGCdual} $\rho^{-1}(C)\subseteq\rho^{-1}(C^\perp)\subseteq\rho^{-1}(C)^*$.
\end{proof}
By Lemma \ref{lem:latdisc} the discriminant of $\rho^{-1}(C)$ is
\[
\left.
\begin{array}{ll}
\frac{1}{p^{2N}}(\Delta^Np^{4k})=\Delta^N&d\equiv1~{\rm mod}~4\\
\frac{1}{(2p)^{2N}}(\Delta^Np^{4k})=\left(\frac{\Delta}{4}\right)^N&d\equiv2,3~{\rm mod}~4
\end{array}
\right\}=d^N
\]
 We have
\begin{proposition}\label{dual}
$\rho^{-1}(C)$ is $d-$modular.  
\end{proposition}
\begin{proof}
	We first prove $\frac{1}{\sqrt{d}}\rho^{-1}(C)=\rho^{-1}(C)^*$ as $\ZZ-$modules. 

Take any $\xv\in\frac{1}{\sqrt{d}}\rho^{-1}(C)$, $\xv=\frac{1}{\sqrt{d}}\xv'$ with $\xv'\in\rho^{-1}(C)$.
Take any $\yv\in\rho^{-1}(C)$.

	\textbf{Case 1} For $d\equiv 1~{\rm mod}~4$,
  \[
	 b_\alpha(\xv,\yv)=\frac{1}{p}\sum_{i=1}^N\Tr{x_iy_i}=\frac{1}{p}\sum_{i=1}^N\Tr{\frac{1}{\sqrt{d}}x'_iy_i}.
  \]
  Since $x'_i\in\O_K$, $\frac{1}{\sqrt{d}}x'_i\in\mathcal{D}_K^{-1}$. We have
  \[
	  \Tr{\frac{1}{\sqrt{d}}x'_iy_i}\in\ZZ\Longrightarrow\Tr{\sum_{i=1}^N\frac{1}{\sqrt{d}}x'_iy_i}\in\ZZ.
  \]
  By the same argument as in the proof of Lemma \ref{lem:CdualinGCdual}, we have $\sum_{i=1}^Nx'_iy_i\in(p)$, so $\sum_{i=1}^N\frac{1}{\sqrt{d}}x'_iy_i\in p\mathcal{D}_K^{-1}$.\\
  Since $\sigma_2(p)=p$, $\sigma_2(\mathcal{D}_K^{-1})=\mathcal{D}_K^{-1}$, $\Tr{\sum_{i=1}^N\frac{1}{\sqrt{d}}x'_iy_i}\in p\mathcal{D}_K^{-1}$. We have $\Tr{\sum_{i=1}^N\frac{1}{\sqrt{d}}x'_iy_i}\in p\mathcal{D}_K^{-1}\cap\ZZ$.

  \textbf{Case 2} For $d\equiv 2,3~{\rm mod}~4$,
  \[
	  b_\alpha(\xv,\yv)=\frac{1}{2p}\sum_{i=1}^N\Tr{x_iy_i}=\frac{1}{p}\sum_{i=1}^N\Tr{\frac{1}{2\sqrt{d}}x'_iy_i}.
  \]
  Since $x'_i\in\O_K$, $\frac{1}{2\sqrt{d}}x'_i\in\mathcal{D}_K^{-1}$. We have
  \[
	  \Tr{\frac{1}{2\sqrt{d}}x'_iy_i}\in\ZZ\Longrightarrow\Tr{\sum_{i=1}^N\frac{1}{2\sqrt{d}}x'_iy_i}\in\ZZ.
  \]
  Similarly we have $\sum_{i=1}^Nx'_iy_i\in(p)$, so $\sum_{i=1}^N\frac{1}{2\sqrt{d}}x'_iy_i\in p\mathcal{D}_K^{-1}$.

  Since $\sigma_2(p)=p$, $\sigma_2(\mathcal{D}_K^{-1})=\mathcal{D}_K^{-1}$, $\Tr{\sum_{i=1}^N\frac{1}{2\sqrt{d}}x'_iy_i}\in p\mathcal{D}_K^{-1}$. Hence we have $\Tr{\sum_{i=1}^N\frac{1}{2\sqrt{d}}x'_iy_i}\in p\mathcal{D}_K^{-1}\cap\ZZ$.

  By definition of different, $\O_K\subseteq\mathcal{D}_K^{-1}$, so $p\O_K\subseteq p\mathcal{D}_K^{-1}$, we have $p\mathcal{D}_K^{-1}\cap\ZZ\supseteq\O_K\cap\ZZ=p\ZZ$, which gives $p\mathcal{D}_K^{-1}\cap\ZZ=\ZZ$ or $p\ZZ$.
  But $p\mathcal{D}_K^{-1}\cap\ZZ=\ZZ$ implies $(p)|\mathcal{D}_K$, which is impossible as $p$ is inert.
  We have $p\mathcal{D}_K^{-1}\cap\ZZ=p\ZZ$ and hence $b_\alpha(\xv,\yv)\in\ZZ$.

  We have proved $\frac{1}{\sqrt{d}}\rho^{-1}(C)\subseteq \rho^{-1}(C)^*$.

  On the other hand,
  \begin{eqnarray*}
	  vol(\frac{1}{\sqrt{d}}\rho^{-1}(C))=vol(\rho^{-1}(C))\left|\rho^{-1}(C)\big/\frac{1}{\sqrt{d}}\rho^{-1}(C)\right|=\sqrt{d^N}\left(\frac{1}{\sqrt{d}}\right)^{2N}=d^{-\frac{N}{2}},
  \end{eqnarray*}
  and \cite{Ebeling}
  \[
	  vol(\rho^{-1}(C)^*)=\frac{1}{vol(\rho^{-1}(C))}=\frac{1}{\sqrt{d^N}}=d^{-\frac{N}{2}}.
  \]
  Thus we have $\rho^{-1}(C)^*=\frac{1}{\sqrt{d}}\rho^{-1}(C)$.
  Define
  \begin{eqnarray*}
	  h:(\rho^{-1}(C),b_\alpha)&\to&(\rho^{-1}(C)^*,b_\alpha)\\
    x&\mapsto&\frac{1}{\sqrt{d}}x.
  \end{eqnarray*}
  By the above, $h$ is a $\ZZ-$linear bijection. Take any $\xv,\yv\in\rho^{-1}(C)$,
  \begin{eqnarray*}
	  d\cdot b_\alpha(h(\xv),h(\yv))&=&d\cdot \Tr{\sum_{i=1}^N\alpha h(x)_ih(y)_i}\\
	  &=&d\cdot \Tr{\sum_{i=1}^N\alpha\frac{1}{\sqrt{d}}x_i\frac{1}{\sqrt{d}}y_i}=d\cdot \Tr{\sum_{i=1}^N\alpha\frac{1}{d}x_i y_i}\\
	  &=&\Tr{\sum_{i=1}^N\alpha x_i y_i}=b_\alpha(\xv,\yv).
  \end{eqnarray*}
  The proof is completed.
\end{proof}
%
%
%
\section{Interesting Lattices from Totally Real Quadratic Fields}
\label{sec:realproperties}

The previous section gave generic methods to construct modular lattices, out of which we now would like to find lattices with good properties in terms of minimal norm or secrecy gain.
The following definitions hold for integral lattices. Let thus $(L,b)$ be an integral lattice with generator matrix $M_L$. We further assume that the lattice is embedded in $\RR^n$, and that $b$ is the natural inner product. We will then denote the lattice by $L$ for short.

\begin{definition}\cite{Conway, Ebeling}\label{def:min}
The {\em minimum}, or {\em minimal norm}, of $L$ in $\RR^n$, is
\begin{equation}\label{eqn:latmin}
	\mu_L:=\min\{\|\xv\|^2:\xv\in L\},
\end{equation}
which is the length of the shortest nonzero vector. 
\end{definition}
One motivation to consider the shortest nonzero vector comes from the sphere packing problem~\cite{Conway}, which requires large minimum.
Upper bounds on the minimum has been established for $d=1$ (unimodular lattices) in~\cite{Conway1988, Conway1998, Mallows1975, Quebbemann1998} and for $d\in \{2,3,5,6,7,11,14,15,23\}$~\cite{Rains}.	
$d-$modular lattices which reach those bounds are called {\em extremal}. 
Classification of known extremal lattices is found in \cite{Bachoc, Conway, Nebe, Nebe1996, Nebe2015, Quebbemann1997} and on the on-line table \cite{LatticeWeb}. 

\begin{definition}\cite[Chapter 2]{Conway}\label{def:theser}
Let $\mathbb{H}=\{\tau\in\mathbb{C}:\im\tau>0\}$. For $\tau\in\mathbb{H}$ let $q=e^{2\pi i\tau}$. The {\em theta series} of the lattice $L$ is the function
\begin{equation}\label{eqn:thetaseries}
	\Theta_L(\tau):=\sum_{\xv\in L} q^{\|\xv\|^2} = \sum_{m\in\ZZ}A_m q^m,
\end{equation}
where the second equality holds because we took $L$ to be integral and $A_m=|\{\xv:~\xv\in L,~\|\xv\|^2=m\}|$.
\end{definition}
The coefficient of $q$ in the second term of $\Theta_L$ is called the kissing number of $L$, and the power of $q$ in the second term gives its minimum.
The theta series helps in determining bounds for the minimum \cite{Rains} as well as classifying lattices \cite{Bocherer}. It has also been used recently to define the notion of secrecy gain.

\begin{definition}\label{def:wsg}
	Let $L$ be an $n-$dimensional $d-$modular lattice. The {\em weak secrecy gain} of $L$, denoted by  $\chi_L^W$, is given by \cite{Oggier}:
\begin{equation}\label{eqn:wsg}
	\chi_L^W = \frac{\Theta_{\sqrt[4]{d}\ZZ^n}(\tau)}{\Theta_L(\tau)}, \tau=\frac{i}{\sqrt{d}},
\end{equation}
noting that the volume of a $d-$modular lattice is $vol(L)=d^{\frac{n}{4}}$. 
\end{definition}
The secrecy gain characterizes the amount of confusion that a wiretap lattice code brings~\cite{Oggier}. 
The weak secrecy gain $\chi_L^W$ is conjectured to be the secrecy gain itself. 
This conjecture is still open, but for large classes of unimodular lattices, it is known to be true \cite{Ernvall, Pinchak}. 
This motivates the study of the relationship between $d$ and $\chi^W_L$ for $d-$modular lattices \cite{ITW2014, Lin2013, Lin}. 
Up to now, no clear conclusion has been drawn. 
We will construct some examples of $d-$modular lattices in Section \ref{sec:realproperties} to gain more information regarding this problem.

Consider the $2N-$dimensional $d-$modular lattice $(\rho^{-1}(C),b_\alpha)$ with 
\[
b_\alpha(\xv,\yv)=
\left\{
\begin{array}{ll}
 \frac{1}{p}\sum_{i=1}^N\Tr{x_iy_i}&d\equiv1~{\rm mod}~4\\
\frac{1}{2p}\sum_{i=1}^N\Tr{x_iy_i}&d\equiv2,3~{\rm mod}~4
\end{array}
\right.
\]
obtained from a self-dual code $C\subseteq\FF_{p^2}^N$, where $p$ a prime inert in $K=\QQ(\sqrt{d})$, for $d$ a square free positive integer. A generator matrix $M_C$ is given by (\ref{eq:MC}).

We thus consider next the following properties of those $d-$modular lattices:

$\bullet$ whether the lattice constructed is even or odd; recall that an integral lattice $(L,b)$ is called {\em even} if $b(x,x)\in 2\ZZ$ for all $x\in L$ and {\em odd} otherwise.

$\bullet$ the minimum of the lattice;

$\bullet$ the theta series and secrecy gain of the lattice.

The computations in this section are mostly done by using SAGE \cite{Sage} and Magma \cite{Magma}.
\subsection{Even/Odd Lattices and Minimum}

We will give general results for the first two properties in this subsection.
By observing the Gram matrices, we have the following results
\begin{proposition}\label{prop:even}
	The lattice $(\rho^{-1}(C),b_\alpha)$ is even iff $d\equiv5~{\rm mod}~8$, $p=2$ and the diagonal entries of $I+AA^T$ are elements from $(4)$.
\end{proposition}
\begin{proof}
\textbf{Case 1}: $d\equiv2,3~{\rm mod}~4$, $2$ is always ramified, so $p$ is an odd prime. And we have
\[
MM^T = 
\begin{bmatrix}
1&1\\
1+\sqrt{d}&1-\sqrt{d}
\end{bmatrix}
\begin{bmatrix}
1&1+\sqrt{d}\\
1&1-\sqrt{d}
\end{bmatrix}
=\begin{bmatrix}
2&2\\
2&2+2d
\end{bmatrix}.
\]
The lower right corner of the Gram matrix is given by
\[
\frac{1}{2}I_k\otimes p
\begin{bmatrix}
2&2\\
2&2+2d
\end{bmatrix}
=I_k\otimes
\begin{bmatrix}
p&p\\
p&p(1+d)
\end{bmatrix}.
\]
Hence the lattice is odd.

\textbf{Case 2}: $d\equiv1~{\rm mod}~4$ and $p$ is an odd prime
\[
MM^T = 
\begin{bmatrix}
1&1\\
\frac{1+\sqrt{d}}{2}&\frac{1-\sqrt{d}}{2}
\end{bmatrix}
\begin{bmatrix}
	1&\frac{1+\sqrt{d}}{2}\\
	1&\frac{1-\sqrt{d}}{2}
\end{bmatrix}
=\begin{bmatrix}
2&1\\
1&\frac{d+1}{2}
\end{bmatrix}.
\]
The lower right corner of the Gram matrix is given by
\[
I_k\otimes p
\begin{bmatrix}
2&1\\
1&\frac{d+1}{2}
\end{bmatrix}
=I_k\otimes
\begin{bmatrix}
2p&p\\
p&p\frac{(1+d)}{2}
\end{bmatrix}.
\]
Hence the lattice is odd.

\textbf{Case 3}: When $d\equiv1~{\rm mod}~4$ and $p=2$, $\O_K=\ZZ[\frac{1+\sqrt{d}}{2}]$. The minimum polynomial of $\frac{1+\sqrt{d}}{2}$ is $f(x) = x^2-x+\frac{1-d}{4}$. We have
\[
	f(x)\equiv
	\begin{cases}
		x^2-x\equiv x(x-1)~{\rm mod}~2&d\equiv1~{\rm mod}~8\\
		x^2-x+1~{\rm mod}~2&d\equiv5~{\rm mod}~8
	\end{cases}
\]
So $2$ is inert only when $d\equiv5~{\rm mod}~8$. In this case, the right lower corner of the Gram Matrix is
\[
I_k\otimes2
\begin{bmatrix}
2&1\\
1&\frac{d+1}{2}
\end{bmatrix}
=
I_k\otimes
\begin{bmatrix}
4&2\\
2&d+1
\end{bmatrix},
\]
which has even diagonal entries.

Furthermore, 
\[
M_1M_1^T=
\begin{bmatrix}
1\\
\frac{1+\sqrt{d}}{2}
\end{bmatrix}
\begin{bmatrix}
	1&\frac{1+\sqrt{d}}{2}
\end{bmatrix}
=
\begin{bmatrix}
	1&\frac{1+\sqrt{d}}{2}\\
	\frac{1+\sqrt{d}}{2}&\frac{1+d+2\sqrt{d}}{4}
\end{bmatrix}.
\]
The left upper corner of the Gram Matrix is
\[
	\frac{1}{2}\Tr{(I+AA^T)\otimes M_1M_1^T}=\Tr{(I+AA^T)\otimes
	\begin{bmatrix}
		\frac{1}{2}&\frac{1+\sqrt{d}}{4}\\
	\frac{1+\sqrt{d}}{4}&\frac{1+d+2\sqrt{d}}{8}	
        \end{bmatrix}
}.
\]
Let $\{\cv_1,\dots,\cv_t\}$ be the rows of $(I ~ A~{\rm mod}~(2))$, i.e. they form a basis for $C$. 
Let $\{\hat{\cv}_1,\dots,\hat{\cv}_k\}$ denote the rows of $(I~A)$, i.e. $\hat{\cv}_i$ is a preimage of $\cv_i$. 
Then the diagonal entries of $\frac{1}{2}\Tr{(I+AA^T)\otimes M_1M_1^T}$ are given by $\Tr{\frac{1}{2}\hat{\cv}_i\cdot\hat{\cv}_i}\text{ and }\Tr{\frac{1+d+2\sqrt{d}}{8}\hat{\cv}_i\cdot\hat{\cv}_i}$.
The lattice is even iff $\forall i$, $\Tr{\frac{1}{2}\hat{\cv}_i\cdot\hat{\cv}_i},\Tr{\frac{1+d+2\sqrt{d}}{8}\hat{\cv}_i\cdot\hat{\cv}_i}\in 2\ZZ$, i.e.
\begin{equation}\label{eqn:cici}
	\Tr{\hat{\cv}_i\cdot\hat{\cv}_i},\Tr{\frac{1+d+2\sqrt{d}}{4}\hat{\cv}_i\cdot\hat{\cv}_i}\in 4\ZZ.
\end{equation}
As $d\equiv1~{\rm mod}~4$,
\[
	\frac{1+\sqrt{d}}{2}=\frac{1+d+2\sqrt{d}}{4}-\frac{d-1}{4}
\]
shows $\{1,\frac{1+d+2\sqrt{d}}{4}\}$ is a $\ZZ-$basis for $\Oc_K$.
Then (\ref{eqn:cici}) is equivalent as $\hat{\cv}_i\cdot\hat{\cv}_i\in4\D^{-1}_K$.
Since $\hat{\cv}_i\in\O_K$, $\hat{\cv}_i\cdot\hat{\cv}_i\in\O_K$, the lattice is even iff 
\[
\hat{\cv}_i\cdot\hat{\cv}_i\in4\D^{-1}_K\cap\O_K=4\O_K
\]
$\hat{\cv}_i\cdot\hat{\cv}_i$ are exactly the diagonal entries of $I+AA^T$ and the proof is completed.
\end{proof}
Next we look at the minimum of some of those lattices.

Consider $d\equiv2,3~{\rm mod}~4$. Let $p$ be a prime such that $\left(\frac{d}{p}\right)=-1$, hence $p$ is inert in $\QQ(\sqrt{d})$ and the finite field $\FF_{p^2}\cong\FF_p(\omega)$, where $\omega$ satisfies the polynomial $x^2-d=0~{\rm mod}~p$. 
Let $C\subseteq\FF_{p^2}^N$ be a self-dual linear code. 
Then each codeword $\cv\in C$ can be written as $\sv + \tv\omega$ for some $\sv,\tv\in\FF_p^N$.
For each coordinate of $\cv$, we have $c_i=s_i+t_i\omega$. Note that each $\hat{\cv}\in\O_K^N$ with $\hat{c}_i=s_i+t_i\sqrt{d}\in\O_K$, where $s_i$ and $t_i$ are considered as integers, is a preimage of $\cv$. 
Furthermore, we can assume $s_i,t_i\in\{-\frac{p-1}{2},\frac{p-3}{2},\dots,-1,0,1,\dots,\frac{p-1}{2}\}$. 
We have proved that $(\rho^{-1}(C),b_\alpha)$ is an odd $d-$modular lattice of dimension $2N$. Moreover, we have
\begin{lemma}
	For $d\equiv2,3~{\rm mod}~4$, the minimum of $(\rho^{-1}(C),b_\alpha)$ is given by
\[
	\min\left\{p,\ \ \min_{\cv\in C\backslash\{\bf{0}\}}b_\alpha(\hat{\cv},\hat{\cv})\right\}
\]
\end{lemma}
\begin{proof}
	Take any $\cv\in C$, then any $\xv\in\rho^{-1}(\cv)$ is of the form $\xv=\hat{\cv}+p\yv$ for some $\yv\in\O_K^N$. Write $y_i=a_i+b_i\sqrt{d}$ for $a_i,b_i\in\ZZ$. Then
	\[
		x_i^2 = (c_i+py_i)^2 = ((s_i+pa_i)+(t_i+pb_i)\sqrt{d})^2, \Tr {x_i^2} = 2(s_i+pa_i)^2+2(t_i+pb_i)^2d.
	\]
	Since $a_i\in\ZZ$, $s_i\in\{-\frac{p-1}{2},\frac{p-3}{2},\dots,-1,0,1,\dots,\frac{p-1}{2}\}$, the minimum value for $(s_i+pa_i)^2$ is $s_i^2$. Similarly, the minimum value for $(t_i+pb_i)^2$ is $t_i^2$. 

	For $\cv\neq\boldsymbol{0}$, minimum value for $\Tr{x_i^2}$ is $2s_i^2+2t_i^2d$ and we have 
	\[
		\min_{\xv\in\rho^{-1}(\cv)}b_\alpha(\xv,\xv)=\frac{1}{2p}\sum_{i=1}^N2(s_i^2+t_i^2d)=\frac{1}{2p}\sum_{i=1}^N\Tr{\hat{c}_i^2}=b_\alpha(\hat{\cv},\hat{\cv}).
	\]
        When $\cv=\bf{0}$
	\[
		b_\alpha(\xv,\xv) = \frac{1}{2p}\sum_{i=1}^N\Tr{x_i^2} = p\sum_{i=1}^N(a_i^2+b_i^2d),
	\]
	which has minimum value $p$ ($\xv\neq\bf{0}$).

	We have
	\[
		\min_{\xv\in\rho^{-1}(C)}b_\alpha(\xv,\xv)=\min_{\cv\in C}\{\min_{\xv\in\rho^{-1}(\cv)}b_\alpha(\xv,\xv)\}=\min\left\{p,\ \ \min_{\cv\in C\backslash\{\bf{0}\}}b_\alpha(\hat{\cv},\hat{\cv})\right\}.
	\]
\end{proof}
\subsection{A New Extremal Lattice.}

	Take $d=5$, $p=2$, $N=6$, $C\subseteq\FF_{4}^6$ with generator matrix $(I~A~{\rm mod}~(2))$ and $A~{\rm mod}~(2)$ is given by
      \[
		\begin{bmatrix}
                  \omega&1&\omega\\
		  0&\omega+1&\omega\\
	 	  \omega+1&\omega+1&1
		\end{bmatrix},
	\]
	where $\omega\in\FF_{4}$ satisfies $\omega^2+\omega+1=0$. Taking $\frac{1+\sqrt{5}}{2}$ to be the preimage of $\omega$, we have $(I+AA^T)_{11} = \sqrt{5}+5\notin(4)$. 
	By Proposition \ref{prop:even} $L$ is an odd $5$-modular lattice of dimension $12$. 
	We computed that this lattice has minimum $4$, kissing number 60. 
	It is an extremal $5$-modular lattice in dimension $12$. It seems to be new.

\subsection{Construction of Existing Lattices}
We present a construction from codes of some well known lattices.
\begin{example}
	Take $d=5$, $p=2$, $N=4$, $C\subseteq\FF_{4}^4$ with generator matrix $(I~A~{\rm mod}~(2))$ and $A~{\rm mod}~(2)$ is given by
    	\[
		\begin{bmatrix}
\omega^2&\omega\\
-\omega&\omega^2
		\end{bmatrix},
	\]
	where $\omega\in\FF_{4}$ satisfies $\omega^2+\omega+1=0$. 
	Taking $\frac{1+\sqrt{5}}{2}$ to be the preimage of $\omega$, we have
	\[
I+AA^T=
\begin{bmatrix}
	2\sqrt{d}+6&0\\
	0&2\sqrt{d}+6
\end{bmatrix}
	\]
	By Proposition \ref{prop:even} $L$ is an even $5$-modular lattice of dimension $8$. 
	This lattice is actually the unique $5$-modular odd lattice of dimension $8$ and minimum $4$ ($Q_8(1)$ in Table 1 of \cite{Rains}).
\end{example}
\begin{example}
	Take $d=6$, $p=7$, $N=4$, $C\subseteq\FF_{25}^4$ with generator matrix $(I~A~{\rm mod}~(7))$ and $A~{\rm mod}~(7)$ is given by
    	\[
		\begin{bmatrix}
2+\omega&2-\omega\\
2-\omega&-2-\omega
		\end{bmatrix},
	\]	
	where $\omega\in\FF_{25}$ satisfies $\omega^2=2$. Then we get the unique $6$-modular odd lattice of dimension $8$ and minimum $3$ ($O^{(6)}$ in Table 1 of \cite{Rains}).
\end{example}
\begin{example}
Take $d=3$, $p=5$, $N=6$, linear code $C\subset\FF_{25}^6$ generated by $(I~A~{\rm mod}~(5))$ and $A~{\rm mod}~(5))$ is given by
\[
		\begin{bmatrix}
\omega+1&2\omega+2&2\\
2\omega+1&2&-\omega+2\\
-\omega+3&\omega+1&2\omega+1
		\end{bmatrix}.
	\]
	We get the unique $3-$modular odd lattice of dimension $12$ and minimum $3$ ($O^{(3)}$ in Table 1 of \cite{Rains}).	
\end{example}

\begin{example}
		Take $d=2$, $p=5$, $N=8$, $C\subseteq\FF_{25}^8$ with generator matrix $(I~A~{\rm mod}~(5))$. $A~{\rm mod}~(5)$ is given by
	      \[
		\left(\begin{array}{rrrr}
2 \omega + 1 & 4 \omega + 1 & 4 \omega + 3 & 4 \omega + 3 \\
\omega + 3 & 2 & 0 & 3 \omega + 4 \\
3 \omega + 3 & 0 & 2 & 4 a + 4 \\
3 \omega + 2 & 3 \omega + 2 & 3 \omega + 1 & 1
\end{array}\right)	\]
	where $\omega\in\FF_{25}$ satisfies $\omega^2=2$. 
	Taking $\sqrt{2}$ to be the preimage of $\omega$, we have the unique odd $2$-modular lattice of dimension $16$ with minimum $3$ (Odd Barnes-Wall lattice $O^{(2)}$ in Table 1 of  \cite{Rains}).
\end{example}
\subsection{Some Lattices with Large Minimum}

We next present two lattices which, though not extremal, achieve a large minimum.
``Large'' means close to the bound that extremal lattices achieve.

We also compute their theta series, so we can later on compute their weak secrecy gain.
Evidence from unimodular lattices~\cite{Lin2013} suggests that a large minimum induces a large weak secrecy gain.

\begin{example}
	Take $d=7$, $p=5$, $N=4$, we get an $8-$dimensional $7-$modular lattice with minimum $3$ and theta series $1 + 16q^3 + 16q^4 + O(q^5)$.
	Note that the upper bound for the minimum of a $8-$dimensional $7-$modular lattice is $4$.
\end{example}
\begin{example}
	Take $d=6$, $N=6$, we can get three different $12-$dimensional $6-$modular lattices with minimum $3$, which is close to the upper bound $4$. Their theta series are as follows:
	\begin{equation*}
		\begin{aligned}
		&1 + 4q^3 + 36q^4 + O(q^5)\\
		&1 + 12q^3 + 40q^4 + O(q^{5})\\
		&1 + 16q^3 + 36q^4 + O(q^{5}).
	\end{aligned}
	\end{equation*}
\end{example}
%
%
\begin{table}
\caption{Weak Secrecy Gain-Dimension 8}
\label{SGTableD8}
\begin{adjustbox}{center}
	
	\small
\begin{tabular}{|c|r|r|c|c|r|rrrrrrrrrr|} 
  \hline			
  No. & Dim & $d$ & $\mu_{L}$ & ks   & $\chi_{L}^W$ &   &   & &  &  $\Theta_{L}$   &   &   &   &   &     \\ 
  \hline
    1	  &  8  &  3  &   2   & 8    &   1.2077  & 1 & 0  & 8  & 64  & 120 & 192  & 424  & 576   & 920   & 1600\\
    2	  &  8  &  5  &   2   & 8    &   1.0020  & 1 & 0  & 8  & 16  & 24  & 96   & 128  & 208   & 408   & 480\\  
    3	  &  8  &  5  &   4   & 120  &   1.2970  & 1 & 0  & 0  & 0   & 120 & 0    & 240  & 0     & 600   & 0 \\
    4	  &  8  &  6  &   3   & 16   &   1.1753  & 1 & 0  & 0  & 16  & 24  & 48   & 128  & 144   & 216   & 400 \\ 
    5	  &  8  &  7  &   2   & 8    &   0.8838  & 1 & 0  & 8  & 0   & 24  & 64   & 32   & 128   & 120   & 192\\ 
    6	  &  8  &  7  &   3   & 16   &   1.1048  & 1 & 0  & 0  & 16  & 16  & 16   & 80   & 128   & 224   & 288\\  
    7	  &  8  &  11 &   3   & 8    &   1.0015  & 1 & 0  & 0  & 8   & 8   & 8    & 24   & 48    & 72    & 88\\
    8	  &  8  &  14 &   2   & 8    &   0.5303  & 1 & 0  & 8  & 0   & 24  & 0    & 32   & 8     & 24    & 64\\  
    9	  &  8  &  14 &   3   & 8    &	 0.9216  & 1 & 0  & 0  & 8   & 0   & 8    & 32   & 0     & 48    & 80\\ 
   10	  &  8  &  15 &   3   & 8    &	 0.8869  & 1 & 0  & 0  & 8   & 0   & 8    & 24   & 0     & 64    & 32 \\
   11	  &  8  &  15 &   4   & 8    &	 1.0840  & 1 & 0  & 0  & 0   & 8   & 16   & 0    & 16    & 32    & 64 \\
   12	  &  8  &  23 &   3   & 8    &	 0.6847  & 1 & 0  & 0  & 8   & 0   & 0    & 24   & 0     & 8     & 40 \\
   13	  &  8  &  23 &   5   & 16   &	 1.0396  & 1 & 0  & 0  & 0   & 0   & 16   & 0    & 0     & 16    & 0 \\
   14	  &  8  &  23 &   5   & 8    &	 1.1394  & 1 & 0  & 0  & 0   & 0   & 8    & 0    & 8     & 24    & 24  \\
\hline 
\end{tabular}

\end{adjustbox}
\end{table}

\subsection{Modular Lattices and their Weak Secrecy Gain}

We are now interested in the relationship between the level $d$ and the weak secrecy gain $\chi^W_{L}$ (see Definition \ref{def:wsg}).
We list the weak secrecy gain of some lattices we have constructed for dimensions $8$ (Table \ref{SGTableD8}), $12$ (Table \ref{SGTableD12}) and $16$ (Table \ref{SGTableD16}). 
In the tables, each row corresponds to a lattice $L$

$\bullet$ labeled by `No.';

$\bullet$ in dimension `Dim';

$\bullet$ of level $d$ (i.e., $L$ is $d-$modular);

$\bullet$ with minimum $\mu_{L}$ (the norm of the shortest vector, see Definition \ref{def:min});

$\bullet$ kissing number `ks' (the number of lattice points with minimal norm);

$\bullet$ obtains weak secrecy gain $\chi_{L}^W$ (see Definition \ref{def:wsg}).

Then in the last column we give the first $10$ coefficients of its theta series $\Theta_{L}$ (see Definition \ref{def:theser}).

\begin{table}
\caption{Weak Secrecy Gain-Dimension 12}
\label{SGTableD12}
\begin{adjustbox}{center}
	
	\small
\begin{tabular}{|c|r|r|c|c|r|rrrrrrrrrr|} 
  \hline			
  No. & Dim & $d$ & $\mu_{L}$ & ks   & $\chi_{L}^W$ &   &   & &  &  $\Theta_{L}$   &   &   &   &   &     \\ 
  \hline

   15     &  12 &   3 &   1   & 12   &	 0.4692  & 1 & 12 & 60 & 172 & 396 & 1032 & 2524 & 4704  & 8364  & 17164\\
   16	  &  12 &   3 &   1   & 4    &	 0.8342  & 1 & 4  & 28 & 100 & 332 & 984  & 2236 & 5024  & 9772  & 16516\\
17	  &  12 &   3 &   1   & 4    &	 0.9385  & 1 & 4  & 12 & 100 & 428 & 984  & 2092 & 5024  & 9708  & 16516\\
18	  &  12 &   3 &   2   & 24   &	 1.2012  & 1 & 0  & 24 & 64  & 228 & 960  & 2200 & 5184  & 10524 & 16192\\
19	  &  12 &   3 &   2   & 12   &	 1.3650  & 1 & 0  & 12 & 64  & 300 & 960  & 2092 & 5184  & 10476 & 16192\\
20	  &  12 &   3 &   3   & 64   &	 1.5806  & 1 & 0  & 0  & 64  & 372 & 960  & 1984 & 5184  & 10428 & 16192\\
21	  &  12 &   5 &   2   & 12   &	 1.0030  & 1 & 0  & 12 & 24  & 60  & 240  & 400  & 984   & 2172  & 3440\\
22	  &  12 &   5 &   4   & 60   &	 1.6048  & 1 & 0  & 0  & 0   & 60  & 288  & 520  & 960   & 1980  & 3680\\  
23	  &  12 &   6 &   1   & 12   &	 0.1820  & 1 & 12 & 60 & 160 & 252 & 312  & 556  & 1104  & 1740  & 2796 \\
24	  &  12 &   6 &   1   & 6    &	 0.3845  & 1 & 6  & 20 & 58  & 132 & 236  & 460  & 936   & 1564  & 2478\\
25	  &  12 &   6 &   2   & 8    &	 0.9797  & 1 & 0  & 8  & 20  & 36  & 144  & 264  & 544   & 1244  & 2016\\
26	  &  12 &   6 &   3   & 16   &	 1.3580  & 1 & 0  & 0  & 16  & 36  & 96   & 256  & 624   & 1308  & 2112 \\
27	  &  12 &   6 &   3   & 12   &   1.3974  & 1 & 0  & 0  & 12  & 40  & 100  & 244  & 668   & 1284  & 2076 \\
28	  &  12 &   6 &   3   & 12   &   1.5044  & 1 & 0  & 0  & 4   & 36  & 132  & 256  & 660   & 1308  & 1980\\
29	  &  12 &   7 &   1   & 12   &   0.1452  & 1 & 12 & 60 & 160 & 252 & 312  & 544  & 972   & 1164  & 1596\\
30	  &  12 &   7 &   1   & 4    &   0.4645  & 1 & 4  & 12 & 32  & 60  & 168  & 416  & 580   & 876   & 1684\\
31	  &  12 &   7 &   1   & 4    &   0.5806  & 1 & 4  & 4  & 16  & 84  & 152  & 208  & 580   & 1268  & 1908\\
32	  &  12 &   7 &   2   & 12   &   0.7584  & 1 & 0  & 12 & 16  & 36  & 144  & 112  & 384   & 852   & 1056\\
33	  &  12 &   7 &   2   & 8    &   0.8795  & 1 & 0  & 8  & 16  & 28  & 112  & 160  & 384   & 772   & 1152\\
34	  &  12 &   7 &   3   & 4    &   1.4023  & 1 & 0  & 0  & 4   & 36  & 84   & 64   & 384   & 972   & 1368\\
35	  &  12 &  11 &   1   & 8    &   0.1765  & 1 & 8  & 24 & 36  & 60  & 180  & 356  & 424   & 612   & 1204\\
36	  &  12 &  11 &   1   & 4    &   0.2173  & 1 & 4  & 16 & 48  & 88  & 152  & 204  & 144   & 316   & 772\\
37	  &  12 &  11 &   3   & 12   &   1.0726  & 1 & 0  & 0  & 12  & 0   & 12   & 108  & 72    & 108   & 436\\
38	  &  12 &  14 &   1   & 8    &   0.1331  & 1 & 8  & 24 & 36  & 56  & 148  & 264  & 320   & 544   & 912\\
39	  &  12 &  14 &   1   & 4    &   0.1534  & 1 & 4  & 16 & 48  & 88  & 152  & 204  & 144   & 280   & 628\\
40	  &  12 &  14 &   3   & 12   &   0.9134  & 1 & 0  & 0  & 12  & 0   & 0    & 72   & 48    & 72    & 256\\ 
41	  &  12 &  15 &   1   & 8    &   0.1313  & 1 & 8  & 24 & 32  & 32  & 112  & 292  & 352   & 328   & 744\\
42	  &  12 &  15 &   1   & 4    &   0.3899  & 1 & 4  & 4  & 0   & 12  & 56   & 96   & 80    & 132   & 388\\
43	  &  12 &  15 &   1   & 2    &   0.4661  & 1 & 2  & 0  & 10  & 32  & 30   & 44   & 96    & 128   & 186 \\   
44	  &  12 &  15 &   2   & 6    &   0.5455  & 1 & 0  & 6  & 8   & 4   & 42   & 46   & 74    & 136   & 154\\
45	  &  12 &  15 &   2   & 6    &   0.9217  & 1 & 0  & 2  & 2   & 4   & 24   & 20   & 46    & 100   & 154\\
46	  &  12 &  15 &   3   & 4    &   1.0031  & 1 & 0  & 0  & 4   & 8   & 18   & 28   & 36    & 64    & 104\\
47	  &  12 &  15 &   4   & 4    &   1.3573  & 1 & 0  & 0  & 0   & 4   & 10   & 12   & 48    & 72    & 108\\
48	  &  12 &  15 &   5   & 4    &   1.5265  & 1 & 0  & 0  & 0   & 0   & 4    & 12   & 44    & 108   & 112 \\
49	  &  12 &  23 &   1   & 8    &   0.0698  & 1 & 8  & 24 & 36  & 56  & 144  & 228  & 192   & 316   & 652\\
50	  &  12 &  23 &   1   & 4    &   0.0735  & 1 & 4  & 16 & 48  & 88  & 152  & 204  & 144   & 280   & 628\\
51	  &  12 &  23 &   3   & 12   &   0.5690  & 1 & 0  & 0  & 12  & 0   & 0    & 60   & 0     & 0     & 172\\	  
 \hline
\end{tabular}

\end{adjustbox}

\end{table}
  
\begin{table}
\caption{Weak Secrecy Gain-Dimension 16}
\label{SGTableD16}
\begin{adjustbox}{center}
	
	\small
\begin{tabular}{|c|r|r|c|c|r|rrrrrrrrrr|} 
  \hline			
  No. & Dim & $d$ & $\mu_{L}$ & ks   & $\chi_{L}^W$ &   &   & &  &  $\Theta_{L}$   &   &   &   &   &     \\ 
  \hline

52  	  &  16 &  3  &   2   & 16   &   1.4585  & 1 & 0  & 16 & 128 & 304 & 1408 & 6864 & 19584 & 47600 & 112768\\
53	  &  16 &  3  &   2   & 12   &   1.6669  & 1 & 0  & 12 & 48  & 440 & 1808 & 6332 & 18864 & 47648 & 113968\\
54	  &  16 &  3  &   2   & 8    &   1.7612  & 1 & 0  & 8  & 48  & 416 & 1808 & 6440 & 18864 & 48016 & 113968\\
55	  &  16 &  3  &   2   & 4    &   1.8303  & 1 & 0  & 4  & 64  & 360 & 1728 & 6676 & 19008 & 48448 & 113728\\
56	  &  16 &  5  &   2   & 2    &   1.7671  & 1 & 0  & 2  & 4   & 72  & 216  & 884  & 2452  & 6432  & 14520\\
57	  &  16 &  5  &   4   & 240  &   1.6822  & 1 & 0  & 0  & 0   & 240 & 0    & 480  & 0     & 15600 & 0\\
58	  &  16 &  5  &   4   & 112  &   1.9213  & 1 & 0  & 0  & 0   & 112 & 0    & 1248 & 2048  & 5872  & 16384\\
59	  &  16 &  5  &   4   & 64   &   1.9855  & 1 & 0  & 0  & 0   & 64  & 192  & 864  & 2432  & 6448  & 14656\\
60	  &  16 &  5  &   4   & 48   &   2.0079  & 1 & 0  & 0  & 0   & 48  & 256  & 736  & 2560  & 6640  & 14080\\
61	  &  16 &  6  &   2   & 16   &   0.8582  & 1 & 0  & 16 & 16  & 112 & 256  & 560  & 1792  & 2928  & 7616\\
62	  &  16 &  6  &   3   & 18   &   1.5662  & 1 & 0  & 0  & 18  & 44  & 122  & 392  & 1050  & 2896  & 7126\\
63	  &  16 &  6  &   3   & 8    &   1.7693  & 1 & 0  & 0  & 8   & 32  & 124  & 376  & 1112  & 3000  & 7156\\
64	  &  16 &  6  &   3   & 8    &   1.8272  & 1 & 0  & 0  & 8   & 16  & 120  & 448  & 1128  & 2992  & 7176\\
65	  &  16 &  7  &   3   & 32   &   1.2206  & 1 & 0  & 0  & 32  & 32  & 32   & 416  & 768   & 1216  & 3648\\
66	  &  16 &  7  &   3   & 6    &   1.7604  & 1 & 0  & 0  & 6   & 12  & 74   & 252  & 560   & 1536  & 3968 \\
67	  &  16 &  7  &   3   & 2    &   1.8381  & 1 & 0  & 0  & 2   & 16  & 86   & 212  & 496   & 1556  & 4072\\
68	  &  16 & 11  &   3   & 16   &   1.0985  & 1 & 0  & 0  & 16  & 0   & 16   & 176  & 96    & 192   & 1072 \\
69	  &  16 & 11  &   3   & 16   &   1.1138  & 1 & 0  & 0  & 16  & 0   & 12   & 164  & 100   & 240   & 1092\\
70	  &  16 & 14  &   3   & 16   &   0.8864  & 1 & 0  & 0  & 16  & 0   & 0    & 128  & 64    & 96    & 640 \\
71	  &  16 & 14  &   3   & 16   &   0.8933  & 1 & 0  & 0  & 16  & 0   & 0    & 124  & 52    & 100   & 676\\
72	  &  16 & 15  &   4   & 6    &   1.5187  & 1 & 0  & 0  & 0   & 6   & 10   & 22   & 54    & 78    & 182\\
73	  &  16 & 15  &   4   & 4    &   1.6192  & 1 & 0  & 0  & 0   & 4   & 4    & 34   & 40    & 74    & 182\\
74	  &  16 & 15  &   4   & 4    &   1.7660  & 1 & 0  & 0  & 0   & 4   & 0    & 14   & 24    & 134   & 156 \\
75	  &  16 & 15  &   4   & 2    &   1.8018  & 1 & 0  & 0  & 0   & 2   & 4    & 10   & 38    & 84    & 208 \\
76	  &  16 & 15  &   5   & 4    &   1.9146  & 1 & 0  & 0  & 0   & 0   & 4    & 8    & 26    & 100   & 178\\
77	  &  16 & 15  &   5   & 4    &   1.9344  & 1 & 0  & 0  & 0   & 0   & 4    & 4    & 36    & 74    & 170\\
78	  &  16 & 15  &   5   & 2    &   1.8890  & 1 & 0  & 0  & 0   & 0   & 2    & 16   & 42    & 70    & 160 \\
79	  &  16 & 23  &   3   & 16   &   0.4715  & 1 & 0  & 0  & 16  & 0   & 0    & 112  & 0     & 0     & 464\\
80	  &  16 & 23  &   3   & 16   &   0.4720  & 1 & 0  & 0  & 16  & 0   & 0    & 112  & 0     & 0     & 460\\
\hline
\end{tabular}
\end{adjustbox}
\end{table}

\begin{remark}\label{rem:wsg}
From the tables we have the following observations:

1. When the dimension increases, the weak secrecy gain $\chi^W_{L}$ tends to increase, a behaviour which has been 
proven for unimodular lattices~\cite{Lin2013};

2. Fixing dimension and level $d$, a large minimum is more likely to induce a large $\chi^W_{L}$, which is also consistent with the observations on unimodular lattices~\cite{Lin2013};

3. Fixing dimension, level $d$ and minimum $\mu_{L}$, a smaller kissing number gives a larger $\chi^W_{L}$ (see e.g. rows 13,14; 15,16,17; 73-75). It was shown for unimodular lattices~\cite{Lin2013} that when the dimension $n$ is fixed, $n\leq 23$, the secrecy gain is totally determined by the kissing number, and the lattice with the best secrecy gain is the one with the smallest kissing number;

4. Fixing dimension, minimum $\mu_{L}$, kissing number, a smaller level $d$ gives a bigger $\chi^W_{L}$. 
For example, in Table~\ref{TableMin3KS16} we list some $16-$dimensional lattices obtaining minimum $3$ and kissing number $16$, with $\chi^W_{L}$ in descending order.

5. Lattices with high level $d$ are more likely to have a large minimum, this is more obvious when the dimension increases (see \cite{Rains}, where the upper bounds for bigger $d$ are also bigger for higher dimensions), and results in bigger $\chi_{L}^W$. 
For example, see rows 13,14,48,76-78.
\end{remark}

Some of those observations can be reasoned by calculating the value of $\chi_{L}^W$:
by (\ref{eqn:wsg}) and (\ref{eqn:thetaseries}), take $\tau=\frac{i}{\sqrt{d}}$, the numerator of $\chi_{L}^W$ is given by
\begin{eqnarray*}
	\Theta_{\sqrt[4]{d}\ZZ^n}\left(\frac{i}{\sqrt{d}}\right)&=&\sum_{\xv\in\sqrt[4]{d}\ZZ^n}q^{\|\xv\|^2}=\sum_{\xv\in\ZZ^n}q^{\sqrt{d}\|\xv\|^2}\\
	&=&\sum_{\xv\in\ZZ^n}e^{\pi\cdot i\cdot i \cdot \frac{1}{\sqrt{d}}\cdot\sqrt{d}\|\xv\|^2}=\sum_{\xv\in\ZZ^n}e^{-\pi\|\xv\|^2},
\end{eqnarray*}
which is a constant. The denominator of $\chi_{L}^W$ is given by
\begin{eqnarray*}
	\Theta_{L}\left(\frac{i}{\sqrt{d}}\right)&=&\sum_{\xv\in L}q^{\|\xv\|^2}=\sum_{\xv\in L}e^{i\pi\cdot\frac{i}{\sqrt{d}}}\\
	&=&\sum_{\xv\in L}e^{\frac{-\pi}{\sqrt{d}}\|\xv\|^2} = \sum_{m\in\ZZ}A_m\left(e^{-\frac{\pi}{\sqrt{d}}}\right)^m,
\end{eqnarray*}
where $A_m$ is the number of vectors in $L$ with norm $m$. 
Hence the denominator can be viewed as a polynomial in $e^{-\frac{\pi}{\sqrt{d}}}$, which is less than $1$.
Then the following will be preferable for achieving a large weak secrecy gain.

1. Large minimum, which determines the lowest power of $e^{-\frac{\pi}{\sqrt{d}}}$ in the polynomial.

2. Small value of $A_m$, i.e., small kissing number.

3. Small value of $d$, so that $e^{-\frac{\pi}{\sqrt{d}}}$ is small.

However, from the three tables, the minimum seems to be more dominant than other factors, and as we mentioned in Remark \ref{rem:wsg} point $5$, large $d$ can still be preferable for high dimensions since it may result in large minima.

\begin{table}
\caption{Weak Secrecy Gain-Dimension 16 and minimum 3}
\label{TableMin3KS16}
\begin{adjustbox}{center}
	
	\small
\begin{tabular}{|c|r|r|c|c|r|rrrrrrrrrr|} 
  \hline			
  No. & Dim & $d$ & $\mu_{L}$ & ks   & $\chi_{L}^W$ &   &   & &  &  $\Theta_{L}$   &   &   &   &   &     \\ 
  \hline
69	  &  16 & 11  &   3   & 16   &   1.1138  & 1 & 0  & 0  & 16  & 0   & 12   & 164  & 100   & 240   & 1092\\
68	  &  16 & 11  &   3   & 16   &   1.0985  & 1 & 0  & 0  & 16  & 0   & 16   & 176  & 96    & 192   & 1072 \\
71	  &  16 & 14  &   3   & 16   &   0.8933  & 1 & 0  & 0  & 16  & 0   & 0    & 124  & 52    & 100   & 676\\
70	  &  16 & 14  &   3   & 16   &   0.8864  & 1 & 0  & 0  & 16  & 0   & 0    & 128  & 64    & 96    & 640 \\
80	  &  16 & 23  &   3   & 16   &   0.4720  & 1 & 0  & 0  & 16  & 0   & 0    & 112  & 0     & 0     & 460\\
79	  &  16 & 23  &   3   & 16   &   0.4715  & 1 & 0  & 0  & 16  & 0   & 0    & 112  & 0     & 0     & 464\\
\hline
\end{tabular}
\end{adjustbox}
\end{table}

%
%
%
\section{Imaginary Quadratic Field}\label{sec:imaginary}
Let $d$ be a positive square-free integer.  
Let $K=\QQ(\sqrt{-d})$ be an imaginary quadratic field with Galois group $\{\sigma_1,\sigma_2\}$, where $\sigma_1$ is the identity map and $\sigma_2:\sqrt{-d}\mapsto-\sqrt{-d}$.
The absolute value of the discriminant of $K$, denoted by $\Delta$, is given by \cite{Neukirch}:
\[
\Delta=\left\{
\begin{array}{ll}
  4d&d\equiv1,2~{\rm mod}~4\\
  d&d\equiv3~{\rm mod}~4
\end{array}
\right..
\]
Assume $p\in\ZZ$ is a prime which is totally ramified in $K$ and let $\p$ be the unique $\O_K-$ideal above $p$.
Then $\Oc_K/\p\cong\FF_p$.
Consider the lattice $(\rho^{-1}(C),b_\alpha)$ 
where $C$ is a linear $(N,k)$ code over $\FF_p$.

Let $\alpha=1/p$ when $d\equiv3~{\rm mod}~4$ and let $\alpha=1/2p$ when $d\equiv1,2~{\rm mod}~4$.
Similarly to Section \ref{sec:real}, we will give two proofs that if $C$ is self-orthogonal (i.e., $C\subseteq C^\perp$), then the lattice $(\rho^{-1}(C),b_\alpha)$ is integral and furthermore we will prove that for $C$ self-dual and for $d$ a prime, we get unimodular lattices.

\subsection{Approach I}
By the discussion from Section \ref{sec:matrix}, a generator matrix for $(\rho^{-1}(C),b_\alpha)$ is (see (\ref{eq:MCCM}))
\begin{equation}\label{eqn:MCIM}
	M_C=\sqrt{\alpha}\begin{bmatrix}
I_k\otimes M&A\otimes M\\
\boldsymbol{0}_{nN-nk,nk}&I_{N-k}\otimes M_p
	\end{bmatrix}
\end{equation}
where $(I_k ~(A \mod \p))$ is a generator matrix for $C$, 
\begin{equation}\label{eq:MIm}
	M=\sqrt{2}
\begin{bmatrix}
  1&0\\
  \re v&-\im v
\end{bmatrix},
M_p=\sqrt{2}
\begin{bmatrix}
	\re \omega_1&-\im\omega_1\\
	\re\omega_2&-\im\omega_2
\end{bmatrix}
\end{equation}
with $\{1,v\}$ a $\ZZ-$basis of $\mathcal{O}_K$, $\{\omega_1,\omega_2\}$ a $\ZZ-$basis of $\p$ and
\begin{equation}\label{eqn:vIM}
v=\left\{
	\begin{array}{ll}
		\frac{1+\sqrt{-d}}{2}&d\equiv3~{\rm mod}~4\\
		\sqrt{-d}&d\equiv1,2~{\rm mod}~4
	\end{array}
\right..
\end{equation}
Its Gram matrix is (see (\ref{eq:GCCM}))
\begin{equation}\label{eqn:GCIM}
	G_C=\alpha\begin{bmatrix}
		(I+AA^T)\otimes\Tr{\vv\vv^\dagger}&A\otimes\Tr{\vv\wv^\dagger}\\
		A^T\otimes\Tr{\bar{\wv}\vv^T}&I_{N-k}\otimes\Tr{\wv\wv^\dagger}
	\end{bmatrix},
\end{equation}
where $\vv=[1,v]^T,\wv=[\omega_1,\omega_2]^T$.
\begin{lem}\label{lem:CIntegralIM}
	If $C$ is self-orthogonal, i.e. $C\subseteq C^\perp$, then $(\rho^{-1}(C),b_\alpha)$ is an integral lattice.
\end{lem}
\begin{proof}
	To prove $(\rho^{-1}(C),b_\alpha)$ is integral, it suffices to prove all entries of its Gram matrix $G_C$ in (\ref{eqn:GCIM}) has integral entries.

	Take any $x\in\p$, as $\p$ is the only prime ideal above $p$, we have $\sigma_2(x)\in\p$ and hence $\Tr{x}\in\p\cap\ZZ=p\ZZ$.
	As $\vv\wv^\dagger,\ \bar{\wv}\vv^T,\ \wv\wv^T$ all have entries in $\p$, $\alpha A\otimes\Tr{\vv\wv^\dagger},\alpha A^T\otimes\Tr{\bar{\wv}\vv^T}$ and $\alpha I_{N-k}\otimes\Tr{\wv\wv^T}$ all have entries in $\ZZ$.
	Furthermore, by Lemma \ref{lem:selfdualcode}, as $C$ is self-orthogonal, $I_k+AA^T~{\rm mod}~\p\equiv\boldsymbol{0}~{\rm mod}~p$ and hence $I_k+AA^T$ has entries in $p\ZZ$.
	We have $\alpha(I+AA^T)\otimes\Tr{\vv\vv^\dagger}$ has integral entries.

When $d\equiv1,2~{\rm mod}~4$, $\Tr{x}$ is even for all $x\in\O_K$.
The proof is completed.
\end{proof}
\begin{proposition}
	If $C$ is self-dual and $d=p$ is a prime, the lattice $(\rho^{-1}(C),b_\alpha)$ is unimodular.
\end{proposition}
\begin{proof}
	By Lemma \ref{lem:CIntegralIM}, the lattice $(\rho^{-1}(C),b_\alpha)$ is integral.
	It suffices to prove it has discriminant $1$ \cite{Ebeling}.
	By Lemma \ref{lem:latdisc}, $(\rho^{-1}(C),b_\alpha)$ has discriminant
	\begin{equation*}
		\Delta^Np^{2k}N(\alpha)^N=\begin{cases}d^Nd^N\left(\frac{1}{d^2})\right)^N&d\equiv3~{\rm mod}~4\\
			(4d)^Nd^N\left(\frac{1}{(2d)^2}\right)^N&d\equiv1,2{\rm mod}~4
		\end{cases}=1.
	\end{equation*}
	\end{proof}

\subsection{Approach II}
In this subsection, we consider $C\subseteq\FF_p^N$ a linear code not necessarily having a generator matrix in the standard form.
We will give another proof that the lattice $(\rho^{-1}(C),b_\alpha)$ is integral, where $\alpha=1/p$ if $d\equiv3~{\rm mod}~4$ and $\alpha=1/2p$ if $d\equiv1,2~{\rm mod}~4$.
Thus $b_\alpha$ is the following bilinear form (see (\ref{eqn:bilinearform})):
\[
b_\alpha(\xv,\yv)\mapsto
\left\{
\begin{array}{ll}
	\frac{1}{p}\sum_{i=1}^N\Tr{x_i\bar{y}_i}&d\equiv1~{\rm mod}~4\\
	\frac{1}{2p}\sum_{i=1}^N\Tr{x_i\bar{y}_i}&d\equiv2,3~{\rm mod}~4
\end{array}
\right..
\]
Then the dual of $(\rho^{-1}(C),b_\alpha)$ is given by $(\rho^{-1}(C)^*,b_\alpha)$, where $\rho^{-1}(C)^*:=\{\xv\in K^N:b_\alpha(\xv,\yv)\in\ZZ~\forall \yv\in\rho^{-1}(C)\}$.
We have the following relation between the dual of $(\rho^{-1}(C),b_\alpha)$ and the lattice constructed from the dual of $C$:
\begin{lem}\label{lem:CdualinGCdualIM}
        $(\rho^{-1}(C^\perp),b_\alpha)\subseteq (\rho^{-1}(C)^*,b_\alpha)$.
\end{lem}
\begin{proof}
	Take any $\xv\in\rho^{-1}(C^\perp)$ and $\yv\in\rho^{-1}(C)$, then
\begin{eqnarray*}
\rho(\xv\cdot\yv)&=&\rho\left(\sum_{i=1}^Nx_iy_i\right)=\sum_{i=1}^N\rho(x_i)\rho(y_i)\\
    &=&\rho(\xv)\cdot\rho(\yv)=0\in\FF_p,
  \end{eqnarray*}
       which gives $\xv\cdot\yv\equiv0~{\rm mod}~\p$.

       As $\p$ is totally ramified, by the same argument as in Proposition \ref{prop:selforthint}, $\beta\equiv \bar{\beta}~{\rm mod}~\p$ for all $\beta\in\O_K$.
       Then we can conclude
	\[
		\xv\cdot\bar{\yv}\equiv\xv\cdot\yv~{\rm mod}~\p\Longrightarrow\xv\cdot\bar{\yv}\in\p.
	\]
	As $\p$ is the only prime above $p$, we have $\sigma_2(\xv\cdot\bar{\yv})\in\p$.
	Hence $\Tr{\xv\cdot\bar{\yv}}\in\p\cap\ZZ=p\ZZ$ and
	\[
		b_{\alpha}(\xv,\yv)=\sum_{i=1}^N\Tr{\alpha x_i\bar{y}_i}=\Tr{\alpha\xv\cdot\bar{\yv}}\in\alpha p\ZZ.
	\]
 	In the case $d\equiv2,3~{\rm mod}~4$, any element in $\O_K$ has even trace.
  	In conclusion, we have $b_\alpha(\xv,\yv)\in\ZZ$ and hence $\rho^{-1}(C^\perp)\subseteq\rho^{-1}(C)^*$ by definition.
\end{proof}
\begin{corollary}
Let $C$ be a self-orthogonal linear code, then $(\rho^{-1}(C),b_\alpha)$ is integral.
\end{corollary}
\begin{proof}
As $C$ is self-orthogonal, we have $C\subseteq C^\perp$.
Hence by Lemma \ref{lem:CdualinGCdualIM} $\rho^{-1}(C)\subseteq\rho^{-1}(C^\perp)\subseteq\rho^{-1}(C)^*$.
\end{proof}
\begin{example}
	Take $d=3$, $K=\QQ{\sqrt{-3}}$, linear code $C\subseteq\FF_3^4$ with generator matrix 
	\[
		\begin{bmatrix}
1&0&2&1\\
0&1&2&2
		\end{bmatrix}.
	\]
	$(\rho^{-1}(C),b_\alpha)$ is a unimodular lattice of dimension $8$ with minimum $2$.
	Thus it is the unique extremal $8-$dimensional unimodular lattice $E_8$ \cite{Conway}.
\end{example}

%
%

\end{document}